\documentclass[
reprint,
amsmath,amssymb,
aps,
]{revtex4-1}

\usepackage{dcolumn}
\usepackage{microtype}
\usepackage{subfigure}
\usepackage{hyperref}
\usepackage{url}
\usepackage{amsmath}
\usepackage{amsthm}
\usepackage{amssymb}
\usepackage{graphicx}
\usepackage{tikz}
\usepackage{bm,float}
\usepackage{qcircuit}
\usepackage{booktabs}
\allowdisplaybreaks

\newtheorem{theorem}{Theorem}
\newtheorem{lemma}{Lemma}

\def\>{\rangle}
\def\<{\langle}

\def\poly{{\rm poly}}
\def\O{\mathcal{O}}

\def\Tr{\text{Tr}}
\def\E{\mathbb{E}}
\def\CX{\text{CX}}
\def\CZ{\text{CZ}}

\begin{document}

\title{Toward Trainability of Deep Quantum Neural Networks}

\author{Kaining Zhang$^1$}
\email{kzha3670@uni.sydney.edu.au}
\author{Min-Hsiu Hsieh$^2$}
\email{min-hsiu.hsieh@foxconn.com}
\author{Liu Liu$^1$}
\email{liu.liu1@sydney.edu.au}
\author{Dacheng Tao$^{3}$}
\email{dacheng.tao@sydney.edu.au}
\affiliation{$^1$School of Computer Science, Faculty of Engineering, University of Sydney, Australia}
\affiliation{$^2$Hon Hai Quantum Computing Research Center, Taipei, Taiwan}

\begin{abstract}
Quantum Neural Networks (QNNs) with random structures have poor trainability due to the exponentially vanishing gradient as the circuit depth and the qubit number increase. This result leads to a general belief that a deep QNN will not be feasible.  In this work, we provide the first viable solution to the vanishing gradient problem for deep QNNs with theoretical guarantees. Specifically, we prove that for circuits with controlled-layer architectures, the expectation of the gradient norm can be lower bounded by a value that is independent of the qubit number and the circuit depth. Our results follow from a careful analysis of the gradient behavior on parameter space consisting of rotation angles, as employed in almost any QNNs, instead of relying on impractical 2-design assumptions.  We explicitly construct examples where only our QNNs are trainable and converge, while others in comparison cannot. 
\end{abstract}

\date{\today}
\maketitle

\section{Introduction} \label{tqnn_intro}

Neural Networks \cite{hecht1992theory} using gradient-based optimizations have dramatically advanced research in discriminative models, generative models, and reinforcement learning. To efficiently utilize parameters and practically improve trainability, neural networks with dedicated architectures \cite{lecun2015deep} are introduced for different tasks, including convolutional neural networks \cite{krizhevsky2012imagenet} for image tasks, recurrent neural networks \cite{zaremba2014recurrent} for time series analysis, and graph neural networks \cite{scarselli2008graph} for tasks related to graph-structured data. Recently, neural architecture search \cite{elsken2019neural} was proposed to improve the performance of the networks by optimizing their neural structures. 

Despite the success in many fields, the development of neural network algorithms is still prone to limitations caused by the large computational resources required for model training. In recent years, quantum computing has emerged as a promising approach to this problem, and has evolved into a new interdisciplinary field known as quantum machine learning (QML) \cite{biamonte2017quantum, havlivcek2019supervised}. Specifically, variational quantum circuits \cite{benedetti2019parameterized} have been explored as efficient protocols for quantum chemistry \cite{kandala2017hardware} and combinatorial optimizations \cite{zhou2018quantum, 2021LeeProgress}. Compared to the classical circuit models, quantum circuits have greater expressive power \cite{Du_2020, shepherd2009temporally} and proven quantum advantage \cite{bravyi2018quantum}. Due to their robustness against noises~\cite{PhysRevResearch.3.023153}, variational quantum circuits have attracted significant interest for the hope of achieving the practical quantum supremacy on near-term quantum computers~\cite{preskill2018quantum, arute2019quantum}.

Quantum Neural Networks (QNNs) \cite{schuld2020circuit, beer2020training} are the special kind of quantum-classical hybrid algorithms that run on trainable quantum circuits. Recently, small-scale QNNs have been implemented on real quantum computers \cite{havlivcek2019supervised, PhysRevApplied.16.024051} for supervised learning tasks. Inspired by the classical optimizations of neural networks, a natural strategy to train QNNs is to exploit the gradient of the loss function \cite{PhysRevA.98.032309, PhysRevLett.118.150503}. However, recent works \cite{mcclean2018barren} show that $N$-qubit quantum circuits with random structures and large depth ($D=\poly(N)$) form approximate unitary $2$-designs \cite{harrow2009random}, and the partial derivative vanishes to zero exponentially with respect to $N$. 

The vanishing gradient problem is usually referred to as the \emph{Barren Plateau} \cite{mcclean2018barren}, which could affect the trainability of QNNs in three folds. 1) Simply using a gradient-based method like Stochastic Gradient Descent (SGD) to train the QNN takes a large number of iterations. 2) The estimation of derivatives requires an extremely large number of samples from the quantum output to guarantee a relatively accurate direction \cite{chen2018gradnorm}. 
3) The objective function tends to have a flat surface \cite{mcclean2018barren} with globally vanishing gradients, which hardly have practical implications.

\subsection{Related Works} \label{tqnn_pre_related_works}

\begin{table*}[t]
\centering
\caption{Existing theoretical results of the barren plateau problem. We use $N$ and $D$ to denote the number of qubits and the circuit depth. We denote by $m$ the number of qubits acted on by local gates. We use $\epsilon$ to measure the distance between circuit states and Haar random states~\cite{holmes2021connecting}. We denote by $q<1$ the noise parameter.}
 \begin{tabular}{|c|c|c|} 
 \hline
 Structure & Distribution & Bounds on the gradient \\ 
 \hline
 Random circuit \cite{mcclean2018barren} & Global $2$-design & $\text{Var} \left[ \frac{\partial f}{\partial \theta_k} \right] \approx \O(2^{-3N})$  \\ 
 Alternating layered circuit \cite{cerezo2020cost} & Local $2$-design & $\text{Var} \left[ \frac{\partial f}{\partial \theta_k} \right] \geq \O(2^{-2m(D+2)})$  \\
QCNN \cite{pesah2020absence} & Local $2$-design & $\text{Var} \left[ \frac{\partial f}{\partial \theta_k} \right] \geq \O(\frac{1}{\poly(N)})$  \\
Arbitrary circuit \cite{holmes2021connecting} & Circuit distribution & $\text{Var} \left[ \frac{\partial f}{\partial \theta_k} \right] \leq 2^{-\O(N)} + \O(\epsilon^2)$ \\
$D$-layer noisy circuit \cite{wang2021noise} & - & $|\frac{\partial f}{\partial \theta_k}| \leq \O(N^{1/2} q^D )$ \\ 
\hline
\end{tabular}
\label{intro_table}
\end{table*}

The barren plateau phenomenon of QNNs is first reported in random ansatzes \cite{mcclean2018barren}. Specifically, for $N$-qubit random quantum circuits with global $2$-design distributed unitaries, the expectation of the derivative to the objective function is zero, and the variance of the derivative vanishes to zero with the rate exponential in the number of qubits. Given the fact that random circuits form approximate $2$-designs~\cite{harrow2009random}, QNNs with random ansatzes are essentially untrainable. Subsequently, Refs.~\cite{cerezo2020cost, 2021Uvarovlocality} relax the global $2$-design assumption to the local case and prove the trainability of shallow alternating layered circuits with local observables. Their results imply that in the low-depth $D=\O(\log N)$ case, the norm of the gradient could be $\frac{1}{\poly (N)}$, which is better than previous exponential vanishing results. Similar results hold for the quantum convolutional neural network (QCNN)~\cite{pesah2020absence} which adopts a tree tensor structure~\cite{zhang2020toward} with $\log N$ depth. Besides, most existing works show a negative relationship between trainability and circuit complexity. Specifically, states generated from QNNs, which satisfy the volume law \cite{PRXQuantum.2.040316}, could lead to the barren plateau landscape. QNNs with high expressibility \cite{holmes2021connecting, PhysRevLett.126.190501} could also suffer from the vanishing gradient problem. For the depolarizing channel case, noisy QNNs with linear depth could have an exponentially small gradient \cite{wang2021noise}. The vanishing gradient phenomenon is inherently related to the concentration nature of measure for quantum states \cite{muller2011concentration}, which cannot be solved using gradient-free \cite{Arrasmith2021effectofbarren} or higher-order \cite{2021CerezoHigher} methods. We summarize existing theoretical results about the barren plateau problem in Table~\ref{intro_table}.

Recently, some techniques have been proposed to address the barren plateau problem, including the special initialization strategy \cite{grant2019initialization} and the layerwise training method \cite{skolik2020layerwise}. We remark that these techniques rely on the assumption of low-depth quantum circuits. Specifically, Ref.~\cite{grant2019initialization} initializes parameters such that the initial quantum circuit is equivalent to an identity matrix ($D=0$). Ref.~\cite{skolik2020layerwise} trains subsets of parameters for each layer so that a low-depth circuit is optimized during the training.

Notice that almost all the theoretical analysis on the trainability of QNNs is based on the assumption that the learning circuits satisfy the unitary $2$-designs or the Haar distributions.  However, QNNs in most scenarios are tuned in the \emph{parameter space} consisting of rotation angles followed by CNOT or CZ gates, instead of the entire unitary space \footnote{Although the assumption of a uniform distribution in the parameter space could generate approximate $2$-designs for certain cases~\cite{larocca2021diagnosing} with linear depths, the corresponding analysis requires that all gates are tunable, which does not hold for circuits employing CNOT or CZ gates.}. Thus, earlier claims that \emph{deep} QNNs are doomed to fail are not on solid ground. In the work, we provide a first analysis of the gradient behavior of the quantum neural network within the parameter space. Most importantly, we theoretically prove the existence of deep QNNs whose gradient norm is independent of the circuit depth $D$.

The rest of this paper is organized as follows. Preliminaries of QNNs are introduced in Section~\ref{tqnn_pre}. The theoretical results are presented in Section~\ref{tqnn_main_results}. We show numerical results in Section~\ref{tqnn_experiments}, which includes {the toy model study (Section~\ref{tqnn_exper_toy}),} finding the ground state of the Ising model (Section~\ref{tqnn_exper_ising}), and the binary classification (Section~\ref{tqnn_exper_qml}). We summarize our conclusion in Section~\ref{tqnn_conclusions}.

\begin{figure*}[t]
\centering
\includegraphics[width=.75\linewidth]{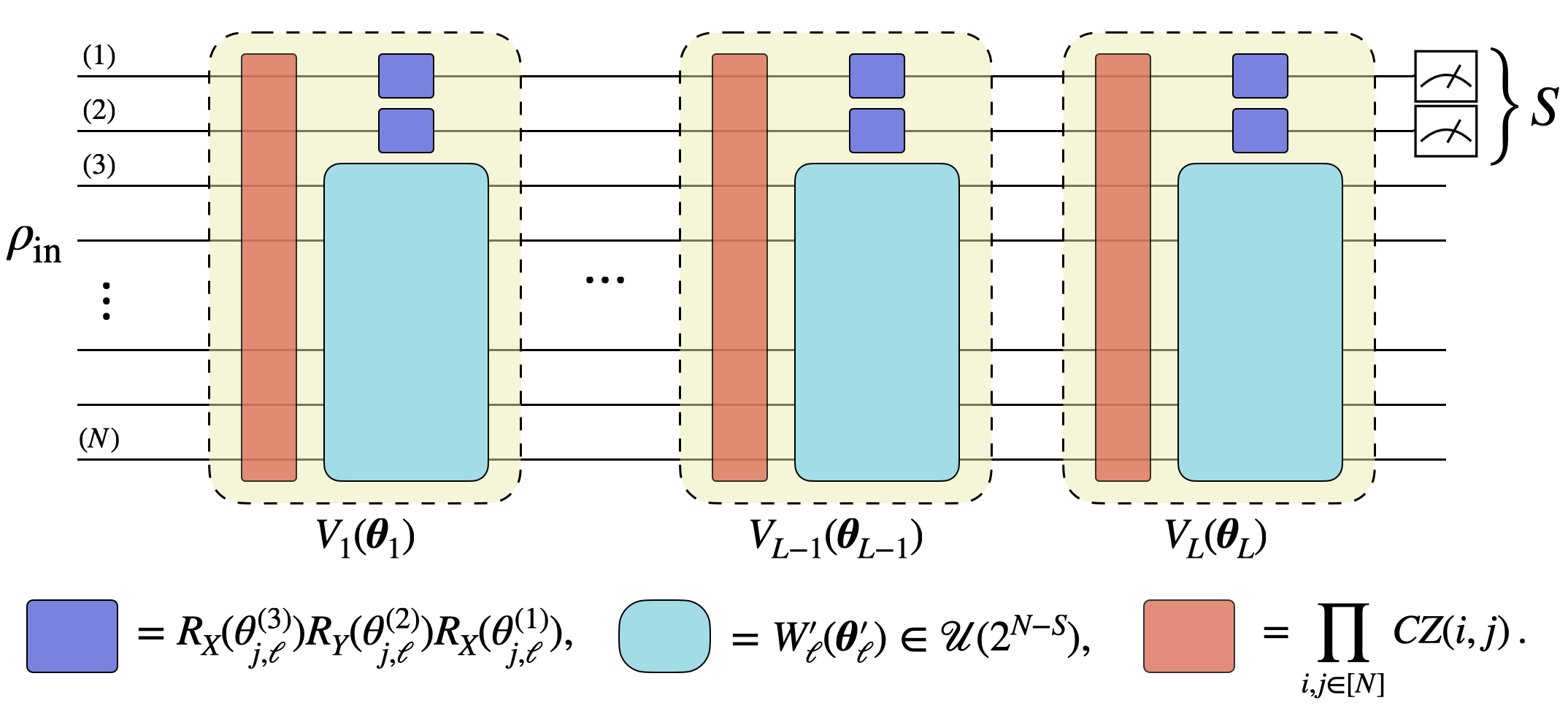}
\caption{The architecture of the $L$ blocks of the controlled-layers quantum neural networks (CL-QNNs)  ($S=2$). Each CL block $V_\ell(\boldsymbol{\theta}_\ell)$ (light yellow block) consists of a CZ operation layer (orange block), single-qubit rotations  {$R_X R_Y R_X$} on the first $S$ qubits (dark blue block), and a parameterized unitary operation $W_\ell'(\boldsymbol{\theta}_\ell')$ on the remaining $N-S$ qubits (light blue block). Operations ${\rm CZ}_\ell$ and $W_\ell'(\boldsymbol{\theta}_\ell')$ could have different structures for different blocks, and the ansatz $W_\ell'(\boldsymbol{\theta}_\ell')$ could be arbitrarily deep. We denote parameters of single-qubit rotations on the $j$th qubit in the $\ell$th CL block by $\boldsymbol{\theta}_{j,\ell} = (\theta_{j,\ell}^{(1)}, \theta_{j,\ell}^{(2)}, \theta_{j,\ell}^{(3)})$. We denote parameters in the $\ell$-th CL block by $\boldsymbol{\theta}_\ell = (\boldsymbol{\theta}_{1,\ell}, \cdots, \boldsymbol{\theta}_{S,\ell}, \boldsymbol{\theta}_\ell ')$, where $\boldsymbol{\theta}_\ell'$ denotes the parameters on the remaining $N-S$ qubits.}
\label{tqnn_main_circuit}
\end{figure*}

\section{Preliminary} \label{tqnn_pre}

Quantum neural network is one special kind of variational quantum algorithms with trainabile parameterized quantum circuits. More specifically, we aim to minimize the function $f$ with respect to $\boldsymbol{\theta}$:  
\begin{equation}\label{tqnn_intro_loss_function}
	f(\boldsymbol{\theta}; \rho_{\text{in}}) = \Tr \left[ O V(\boldsymbol{\theta})  \rho_{\text{in}}  V(\boldsymbol{\theta})^\dag  \right],
\end{equation}
where $O$ denotes the quantum observable and $\rho_{\text{in}}$ denotes the density matrix of the input state. We use the linear combination of tensor products of Pauli matrices
\begin{equation*}
\{\sigma_0, \sigma_1, \sigma_2, \sigma_3\} = \{I, X, Y, Z\}
\end{equation*}
to represent the quantum observable. Equation~(\ref{tqnn_intro_loss_function}) provides a general formulation of the loss in QNNs. For the quantum machine learning (QML) scenario, $\rho_{\text{in}}$ encodes the information of the training data, while the observable $O$ could be a simple $\sigma_3$ for classification tasks~\cite{schuld2020circuit,Grant_2018} or a complex form for quantum kernel methods~\cite{huang2021power,Wang2021towards}. For the variational quantum eigensolver scenario, such as the quantum chemistry and the quantum simulation, $\rho_{\text{in}}$ is usually initialized as $(|0\>\<0|)^{\otimes N}$, while $O$ encodes the Hamiltonian of the whole system which is written as the linear combination of Pauli tensor products. For convenience, we define the \emph{locality} of a Pauli product observable as the number of qubits such that the corresponding component in the observable is not the identity.

Different from the input state $\rho_{\text{in}}$ and the observable $O$, parameters $\boldsymbol{\theta}$ in QNN are trainable, which makes the QNN suitable for various optimizers~\cite{doi:10.1126/sciadv.aaw9918, Stokes2020quantumnatural, Sweke2020stochasticgradient}. In practice, we could deploy the parameter as the phase of single-qubit gates $\{e^{-i \theta \sigma_k}, k \in \{1,2,3\}\}$ while employing two-qubit gates $\{\CX, \CZ\}$ among them to generate quantum entanglement. This strategy has been frequently used in existing quantum neural networks  \cite{schuld2020circuit, benedetti2019parameterized, du2020learnability, du2020quantum, du2021exploring, PhysRevApplied.16.024051, du2021accelerating} since the model suits near-term quantum computers. We remark that the partial derivative could be calculated subsequently by using the parameter shifting rule~\cite{PhysRevA.98.032309, PhysRevLett.118.150503},
\begin{equation}\label{tqnn_gradient_calculate}
\frac{\partial f}{\partial \theta_j} = f(\boldsymbol{\theta}_{+};\rho_{\text{in}}) - f(\boldsymbol{\theta}_{-};\rho_{\text{in}}),
\end{equation}
where $\boldsymbol{\theta}_+$ and $\boldsymbol{\theta}_-$ are different from $\boldsymbol{\theta}$ only at the $j$th parameter: $\theta_j \rightarrow \theta_j \pm \frac{\pi}{4}$. 
Thus, the gradient of $f$ could be obtained by estimating quantum observables, which allows the optimization of QNNs using gradient-based methods~\cite{doi:10.1126/sciadv.aaw9918, Stokes2020quantumnatural, Sweke2020stochasticgradient}.

\section{Main results}
\label{tqnn_main_results}

In this section, we introduce quantum neural networks (QNNs) with controlled-layer (CL) architectures and the corresponding trainability results. Namely, we prove that the expectation of the square of the gradient for the CL-QNN is lower bounded by ${\Omega}(8^{-LS})$, where $L$ is the number of CL blocks in QNNs and $S$ is the locality of the quantum observable. Note that the circuit depth in each CL block is arbitrary. The bounds are independent of the circuit depth, i.e., we guarantee the trainability of arbitrary deep CL-QNNs with limited CL blocks. 

\subsection{Controlled-layer QNNs} \label{tqnn_ttn_method}

\begin{figure*}[t]
 \subfigure[]{
  \centering
  \includegraphics[width=.287692\linewidth]{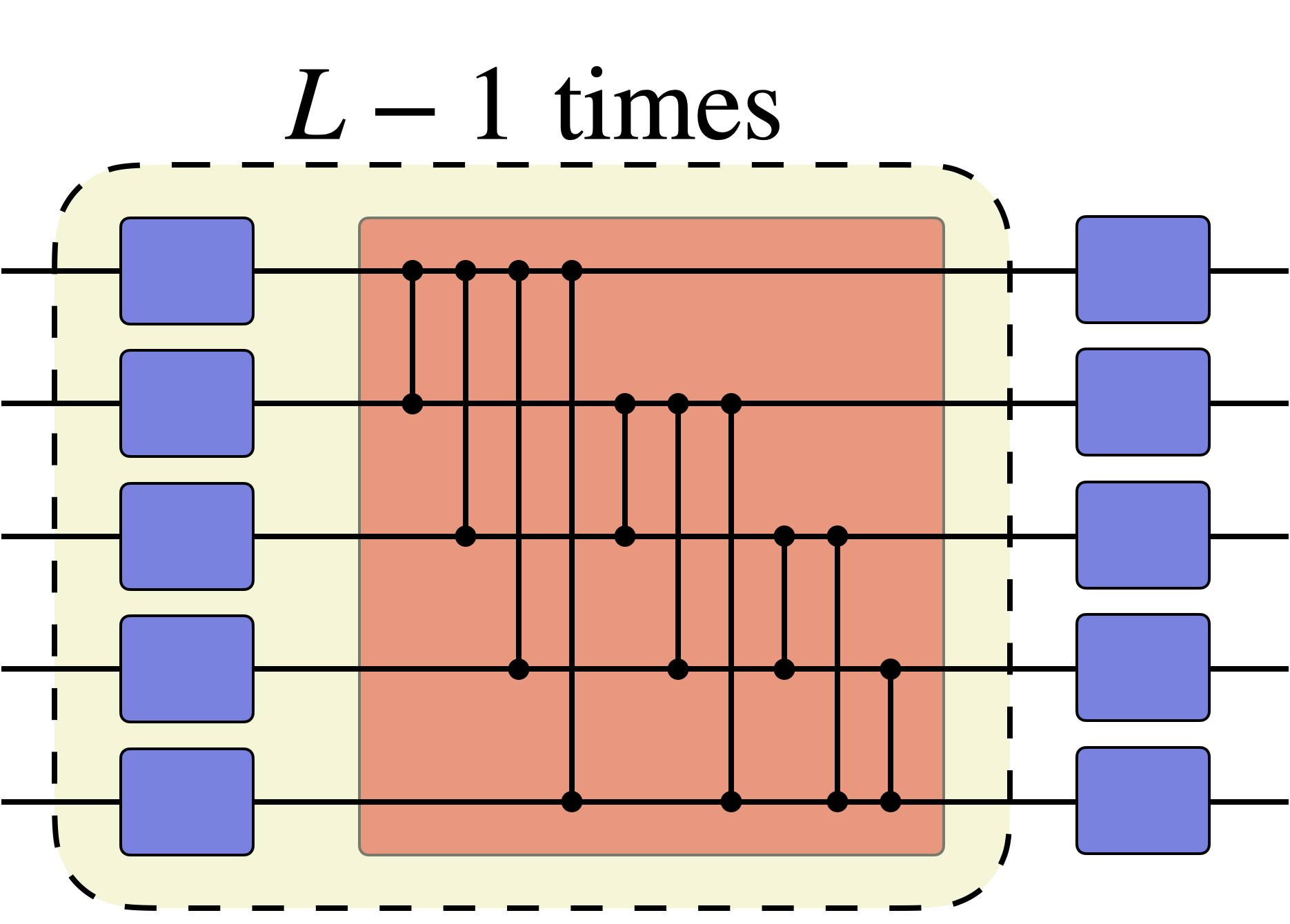}  
  \label{tqnn_ll_circuit}
}
 \subfigure[]{
  \centering
  \includegraphics[width=.206154\linewidth]{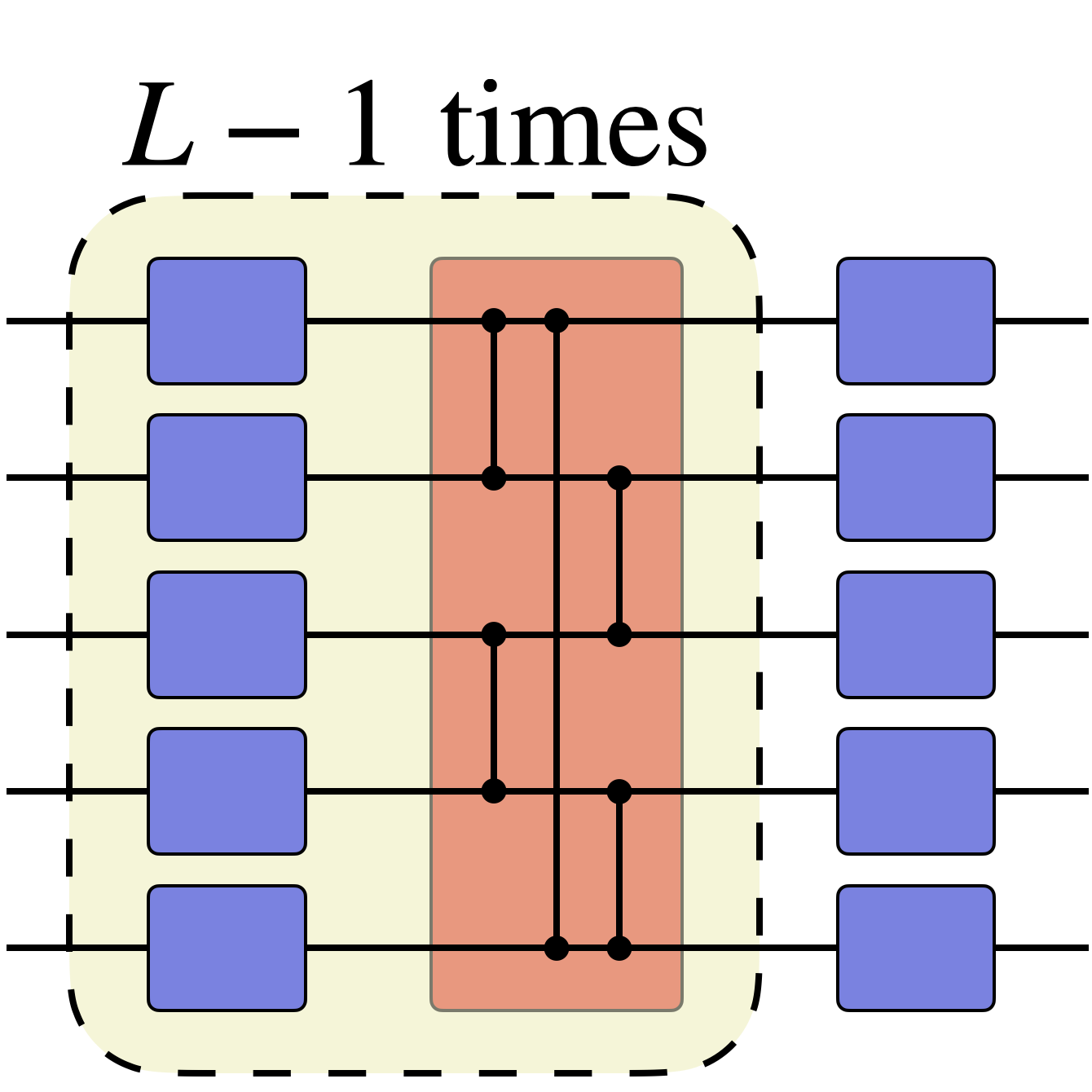}  
  \label{tqnn_qe_circuit}
}
  \subfigure[]{
  \centering
  \includegraphics[width=.3061554\linewidth]{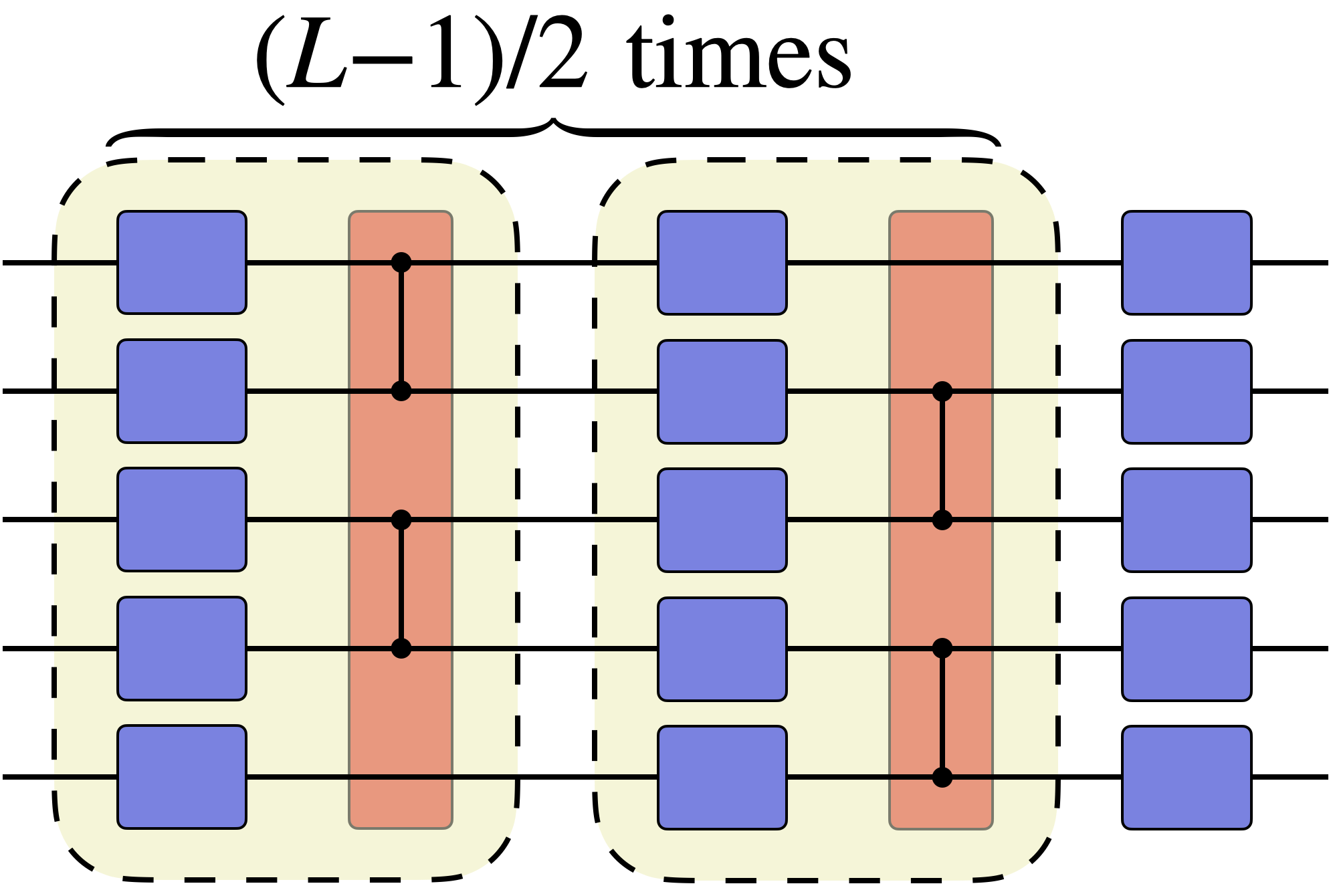}  
  \label{tqnn_al_circuit}
}
\caption{Existing QNNs as special cases of CL-QNNs. Figure~\ref{tqnn_ll_circuit} shows the structure of the layerwise learning circuit~\cite{skolik2020layerwise}. Figure~\ref{tqnn_qe_circuit} shows the structure of the quantum entanglement circuit~\cite{10.1145/3466797, haug2021capacity, haug2021optimal}. Figure~\ref{tqnn_al_circuit} shows the structure of the alternating-layered circuit~\cite{skolik2020layerwise}.  The deep blue block denotes the single-qubit rotations and the orange block denotes the CZ operator layer, as defined in Figure~\ref{tqnn_main_circuit}.
}
\label{tqnn_existing_circuits}
\end{figure*}

Now we discuss the CL-QNNs in detail. The architecture of the CL-QNN is shown in Figure~\ref{tqnn_main_circuit}. The circuit begins with an input state $\rho_{\text{in}}$, followed by $L$ layers of CL blocks. In the $\ell$-th CL block, a CZ operator layer ({denoted as $\text{CZ}_\ell$}), where CZ gates on arbitrary qubit pairs are allowed, is applied. Then we perform parameterized quantum gates on the first $S$ qubits and the remaining $N-S$ qubits separately. For each of the first $S$ qubits, we employ $R_X(\theta_{j,\ell}^{(3)}) R_Y(\theta_{j,\ell}^{(2)}) R_X(\theta_{j,\ell}^{(1)})$ gates sequence to implement arbitrary unitaries in $\mathcal{U}(2)$ individually, where $j\in \{1,\cdots,S\}$ denotes the index of the qubit and $\ell \in \{1,\cdots,L\}$ denotes the index of the CL block. We allow arbitrary parameterized unitaries {$W_\ell ' (\boldsymbol{\theta}_\ell ')\in \mathcal{U}(2^{N-S})$} on the remaining $N-S$ qubits. Thus, the $\ell$-th CL block is defined by
\begin{equation*}
V_\ell(\boldsymbol{\theta}_\ell) = \left( \left( \bigotimes\limits_{j=1}^{S} W_{j,\ell} (\boldsymbol{\theta}_{j,\ell}) \right) \bigotimes W_\ell ' (\boldsymbol{\theta}_\ell ') \right) \text{CZ}_\ell,
\end{equation*}
where $W_{j,\ell} (\boldsymbol{\theta}_{j,\ell}) \equiv R_X(\theta_{j,\ell}^{(3)}) R_Y(\theta_{j,\ell}^{(2)}) R_X(\theta_{j,\ell}^{(1)})$ and $\boldsymbol{\theta}_\ell = (\boldsymbol{\theta}_{1,\ell}, \cdots, \boldsymbol{\theta}_{S,\ell}, \boldsymbol{\theta}_\ell ')$. After all unitary operations, we measure the first $S$ qubits to reveal the information from the quantum circuit. Specifically, the objective function is the expectation of the measurement result with the associated $S$-local observable
\begin{align}
\sigma_{\boldsymbol{i}} &= \sigma_{(i_1, \cdots, i_S, 0, \cdots, 0)} \nonumber \\
&= \sigma_{i_1} \otimes \cdots \otimes \sigma_{i_S} \otimes \sigma_{0} \otimes \cdots \otimes \sigma_{0}, \label{tqnn_observable}
\end{align}
where $i_j \neq 0, \ \forall j \in \{1,\cdots,S\}$.

We remark that the controlled-layer architecture shown in Figure~\ref{tqnn_main_circuit} is a general framework. Arbitrary structures could be employed for operations {${\rm CZ}_\ell$ and $W_\ell'$}, to generate different distributions for specific tasks. Many existing quantum neural networks can be viewed as special cases of CL-QNNs, e.g., the layerwise learning circuit~\cite{skolik2020layerwise} in Figure~\ref{tqnn_ll_circuit}, the quantum entanglement circuit~\cite{10.1145/3466797, haug2021capacity, haug2021optimal} in Figure~\ref{tqnn_qe_circuit}, and the alternating-layered circuit  \cite{cerezo2020cost} in Figure~\ref{tqnn_al_circuit}. Specifically, these architectures contain a CZ  operation layer (orange color) and the arbitrary parameterized unitaries {$W_\ell ' (\boldsymbol{\theta}_\ell ')\in \mathcal{U}(2^{N-S})$} on the remaining $N-S$ qubits are taken to be just single-qubit rotations. Besides, the local quantum observable condition is naturally satisfied for QNNs associated with quantum machine learning~\cite{schuld2020circuit}, combinatorial optimizations~\cite{zhou2018quantum, 2021LeeProgress}, and quantum simulations~\cite{CerveraLierta2018exactisingmodel, Pagano25396}.

\subsection{Trainability of CL-QNNs}

We provide main theoretical results of this work in Theorems~\ref{tqnn_appendix_lemma_cost} and \ref{tqnn_appendix_theorem}. Specifically, we prove that both the square of the loss $f$ and the norm of its gradient are lower bounded by ${{\Omega}}(8^{-LS})$. Both bounds are independent of the qubit number $N$ and the circuit depth $D$, where the latter can be arbitrarily large when employing complex unitaries $W_\ell'$ in Figure~\ref{tqnn_main_circuit}. Proofs of Theorems~\ref{tqnn_appendix_lemma_cost} and \ref{tqnn_appendix_theorem} are provided in Appendices~\ref{tqnn_app_proof_cost} and \ref{tqnn_app_proof_gradient}, respectively. 

\begin{theorem}\label{tqnn_appendix_lemma_cost}
Consider the $N$-qubit $L$-block variational quantum circuit $V(\boldsymbol{\theta})$ defined in Figure~\ref{tqnn_main_circuit} and the cost function $f(\boldsymbol{\theta}) = {\rm Tr} \left[ \sigma_{\boldsymbol{i}} V(\boldsymbol{\theta}) \rho_{\text{in}} V(\boldsymbol{\theta})^{\dag}\right]$, where $\sigma_{\boldsymbol{i}}$ is a $S$-local observable. 
Then the following formula holds,
\begin{equation}\label{tqnn_appendix_lemma_cost_equation}
\mathop{\E}\limits_{\boldsymbol{\theta}} f^2 \geq  \frac{ \left( {\rm Tr}[\sigma_{\boldsymbol{3|i}} \rho_{\rm {in}}] \right)^2 }{8^{LS}},	
\end{equation}
where the index $\boldsymbol{3|i}=(3,3,\cdots,3,0,\cdots,0)$ contains $S$ non-zero components. The expectation is taken independently for all parameters in $\theta$ with the uniform distribution in $[0,2\pi]$.
\end{theorem}

\begin{theorem}\label{tqnn_appendix_theorem}
Consider the $N$-qubit $L$-block variational quantum circuit $V(\boldsymbol{\theta})$ defined in Figure~\ref{tqnn_main_circuit} and the cost function $f(\boldsymbol{\theta}) = {\rm Tr} \left[ \sigma_{\boldsymbol{i}} V(\boldsymbol{\theta}) \rho_{\text{in}} V(\boldsymbol{\theta})^{\dag}\right]$, where $\sigma_{\boldsymbol{i}}$ is a $S$-local observable. Then the following formula holds,
\begin{align}\label{tqnn_appendix_theorem_eq}
\mathop{\E}\limits_{\boldsymbol{\theta}} \|\nabla_{\boldsymbol{\theta}} f \|^2 \geq \frac{ 12(L-1)S }{8^{LS}} \left( {\rm Tr}[\sigma_{\boldsymbol{3|i}} \rho_{\rm {in}}] \right)^2, 
\end{align}
where index $\boldsymbol{3|i}=(3,3,\cdots,3,0,\cdots,0)$ contains $S$ non-zero components. The expectation is taken independently for all parameters in $\theta$ with the uniform distribution in $[0,2\pi]$.
\end{theorem}

We remark that the trace term in Eqs.~(\ref{tqnn_appendix_lemma_cost_equation}) and (\ref{tqnn_appendix_theorem_eq}) is independent of $N$ for many input states, e.g., $\Tr[\sigma_{\boldsymbol{3|i}} (|0\>\<0|)^{\otimes N}]=1$.

\begin{figure}[t]
 \subfigure[]{
  \centering
  \includegraphics[width=.46\linewidth]{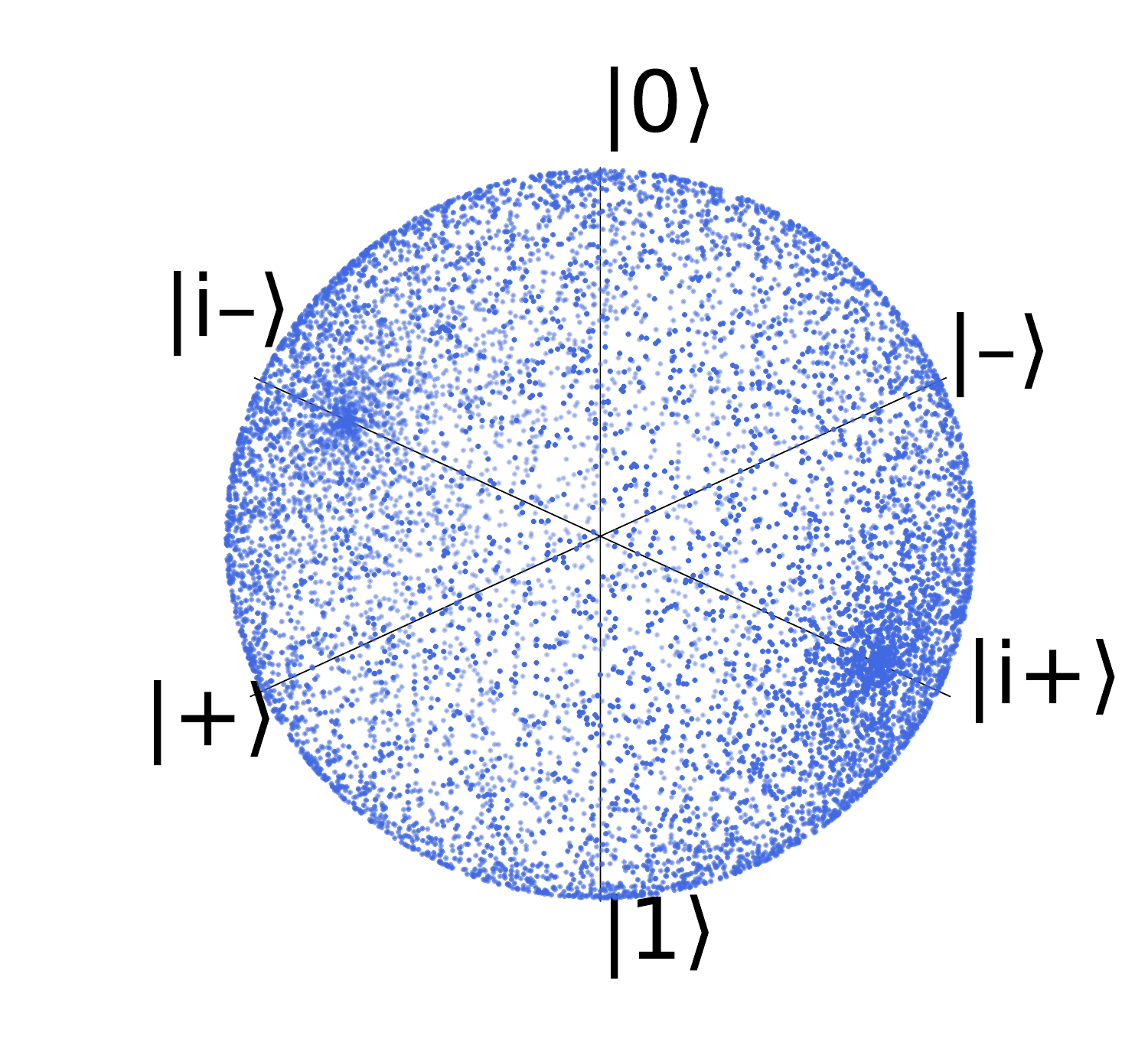}  
  \label{tqnn_fig_dist_uni}
}
 \subfigure[]{
  \centering
  \includegraphics[width=.46\linewidth]{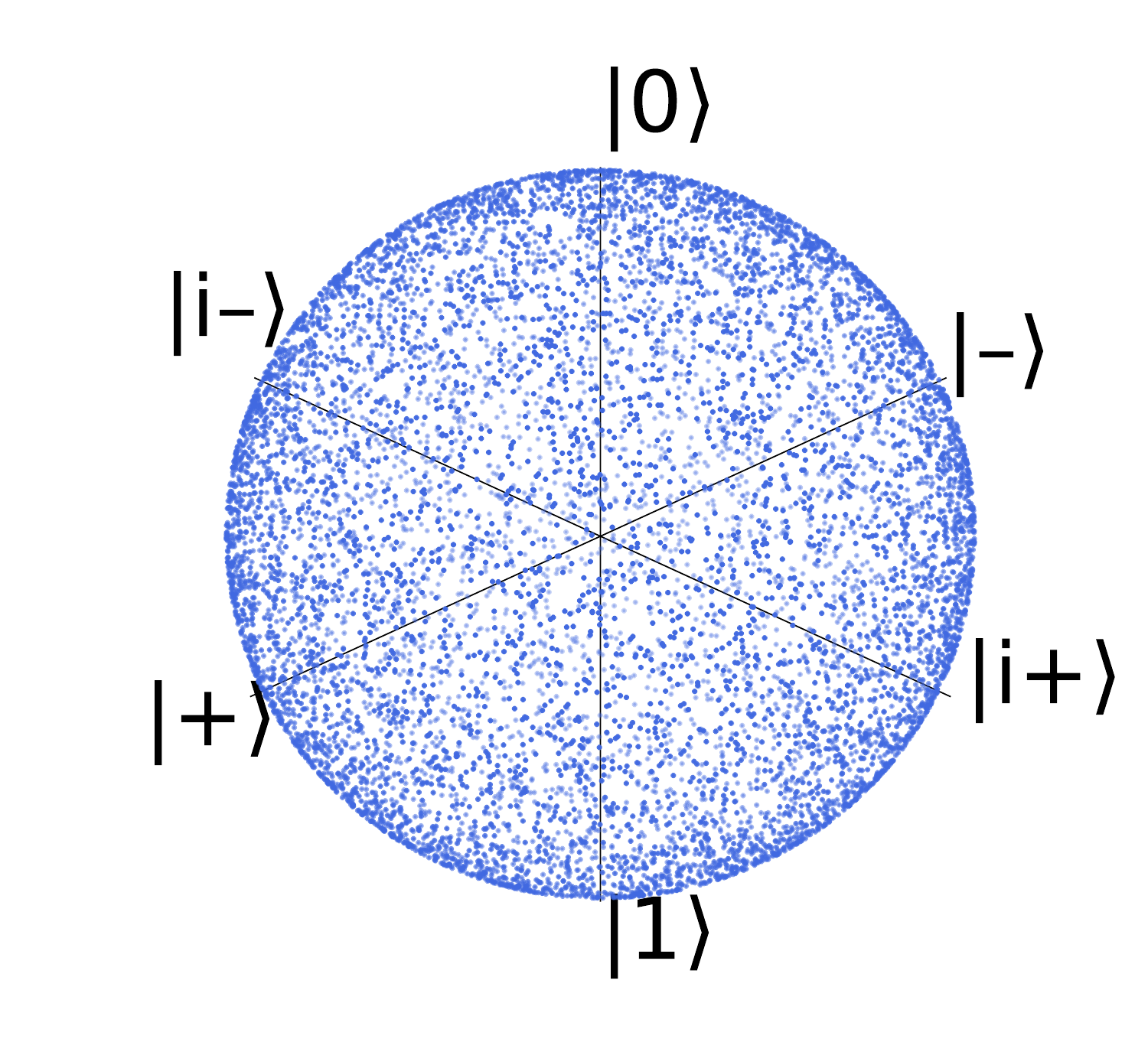}  
  \label{tqnn_fig_dist_haar}
}
\caption{Bloch sphere of states $U|0\>$, where $U=R_Y(\theta_2)R_X(\theta_1)$ for Figure~\ref{tqnn_fig_dist_uni} and $U \in \mathcal{U}(2)$ for Figure~\ref{tqnn_fig_dist_haar}. For the left figure, we sample from the uniform distribution of $(\theta_1, \theta_2) \in [0,2\pi]^{2}$, while for the right figure, we sample from the Haar distribution of the unitary space. Each figure contains 10000 samples.}
\label{tqnn_fig_distributions}
\end{figure}

From the geographic view, the value $\mathop{\E}\limits_{\boldsymbol{\theta}} \|\nabla_{\boldsymbol{\theta}} f\|^2$ characterizes the global steepness of the function surface in the parameter space. Optimizing the objective function $f$ using gradient-based methods could be hard if the norm of the gradient vanishes to zero. Thus, the lower bound in Eq.~(\ref{tqnn_appendix_theorem_eq}) provides a theoretical guarantee for the optimization of the objective function, which is the necessary condition for obtaining a good trainability of QNNs on related tasks. Different from existing works \cite{mcclean2018barren, grant2019initialization, cerezo2020cost, pesah2020absence} that deal with shallow quantum circuits, we provide a positive result on the trainability of deep quantum circuits with certain structures. 

\begin{figure*}[t]
\centering
\includegraphics[width=.75\linewidth]{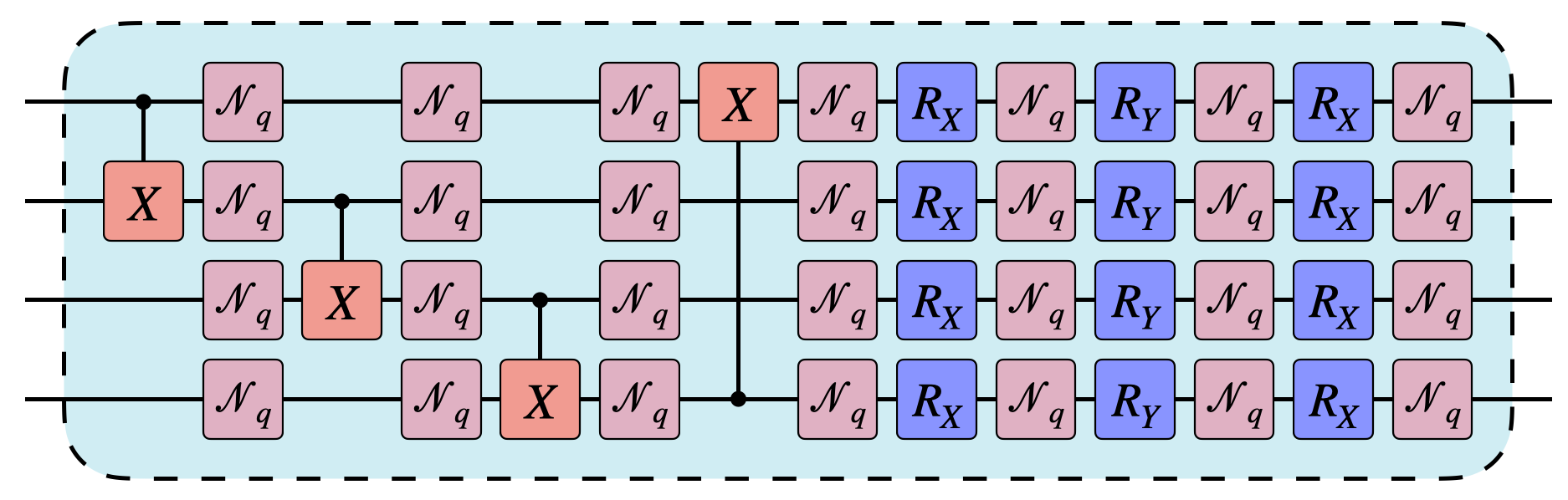}
\caption{{
Illustration of the depolarizing noisy channel used in this work. We denote by $\mathcal{N}_q$ the depolarizing channel on the single qubit, which acts as Eq.~(\ref{tqnn_toy_noisechannel_eq}) on all qubits after each operation layer. Figure~\ref{tqnn_circuit_noise} shows the noisy hardware-efficient ansatz as an example.}}
\label{tqnn_circuit_noise}
\end{figure*}

From the technical view, we provide a new theoretical framework for proving Eqs.~(\ref{tqnn_appendix_lemma_cost_equation}) and (\ref{tqnn_appendix_theorem_eq}). Instead of the unitary $2$-design assumptions which do not accord with practical QNNs training scenarios, we consider the uniform distribution in the parameter space,  in which each parameter in $\boldsymbol{\theta}$ varies continuously in $[0,2\pi]$. Our assumption suits quantum circuits that encode parameters in the phase of single-qubit rotations, and fits the analysis to the loss landscape within the parameter space. Our framework could be extended with further assumptions on current circuit architectures in future works.

Generally, the uniform distribution of parameters could induce a distribution that differs from the Haar distribution~\cite{ozols2009generate}, as illustrated by an example in Figure~\ref{tqnn_fig_distributions}. Specifically, Figure~\ref{tqnn_fig_dist_uni} corresponds to the unitary $U=R_Y(\theta_2)R_X(\theta_1)$, where ${\theta_1, \theta_2}$ are sampled from $[0,2\pi]$ uniformly; while Figure~\ref{tqnn_fig_dist_haar} corresponds to the unitary $U$, which is sampled from the Haar distribution. We can see concentrations of samples around states $|i\pm\> = \frac{|0\> \pm i|1\>}{\sqrt{2}}$ in Figure~\ref{tqnn_fig_dist_uni}, which differs from uniformly distributed samples in Figure~\ref{tqnn_fig_dist_haar}.

\section{Applications} \label{tqnn_experiments}

We provide numerical analysis on the trainability of different QNNs using the PennyLane Python package \cite{bergholm2020pennylane}. We propose {three} optimization tasks, {i.e., the toy model,} quantum simulation of the Ising model, and supervised learning for the binary classification. {In the first and the third task, we are able to show that only CL-QNNs are trainable with deep circuit structures while hardware-efficient QNNs \cite{PMID32874524} and randomly structured QNNs are not. Moreover, In the second task, we demonstrate that the CL-QNNs could perform better than the randomly structured QNNs.}

\subsection{Toy model}
\label{tqnn_exper_toy}

\begin{figure}[t]
\centering
\includegraphics[width=.7\linewidth]{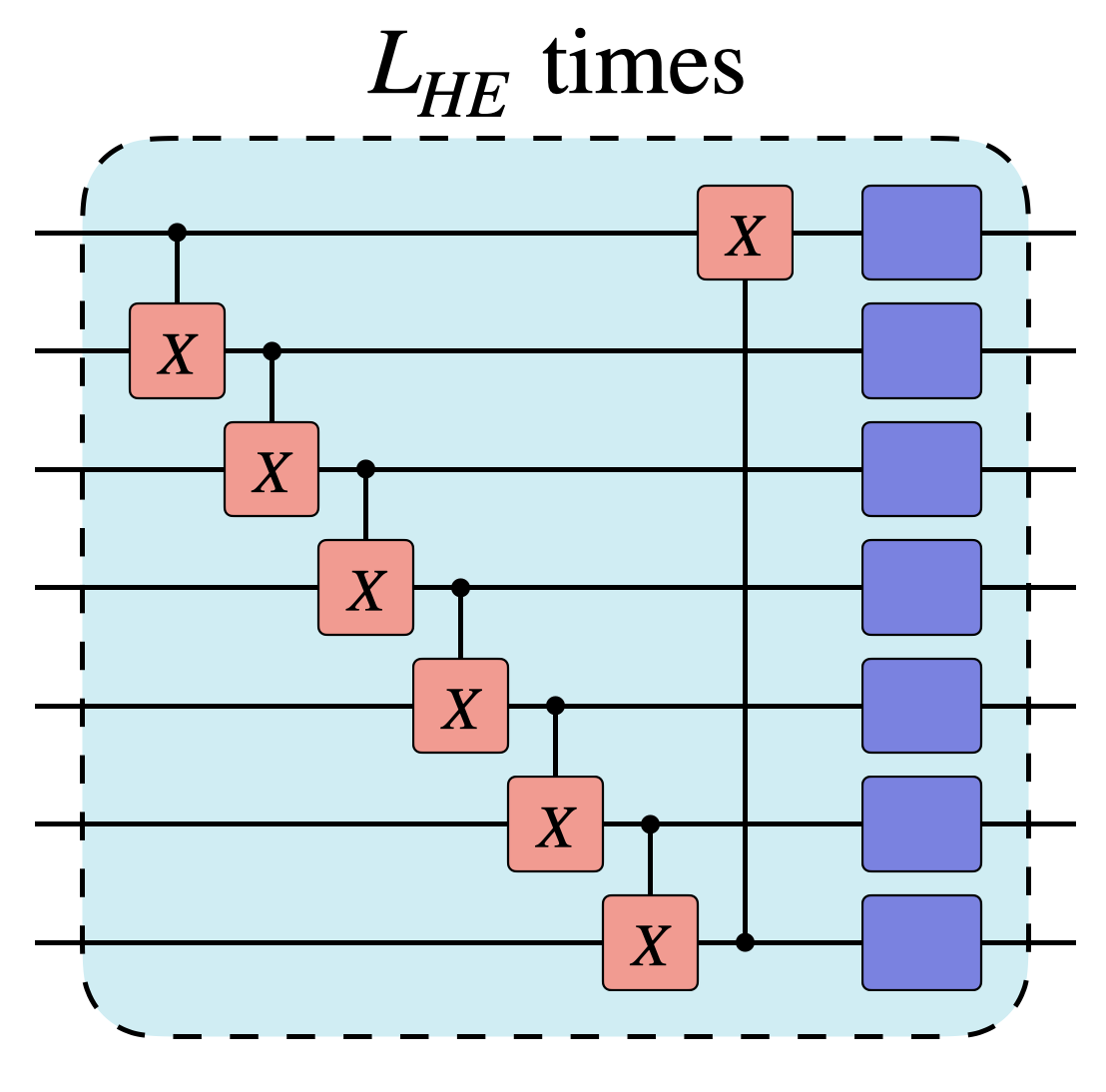}  
\caption{The illustration of the hardware-efficient circuit \cite{PMID32874524}, where the ansatz in the light blue block is repeated for $L_{HE}$ times. The deep blue block denotes the $R_X R_Y R_X$ sequence as defined in Figure~\ref{tqnn_main_circuit}.}
\label{tqnn_he_circuit}
\end{figure}

\begin{figure*}[t]
 \subfigure[]{
  \centering
  \includegraphics[width=.4\linewidth]{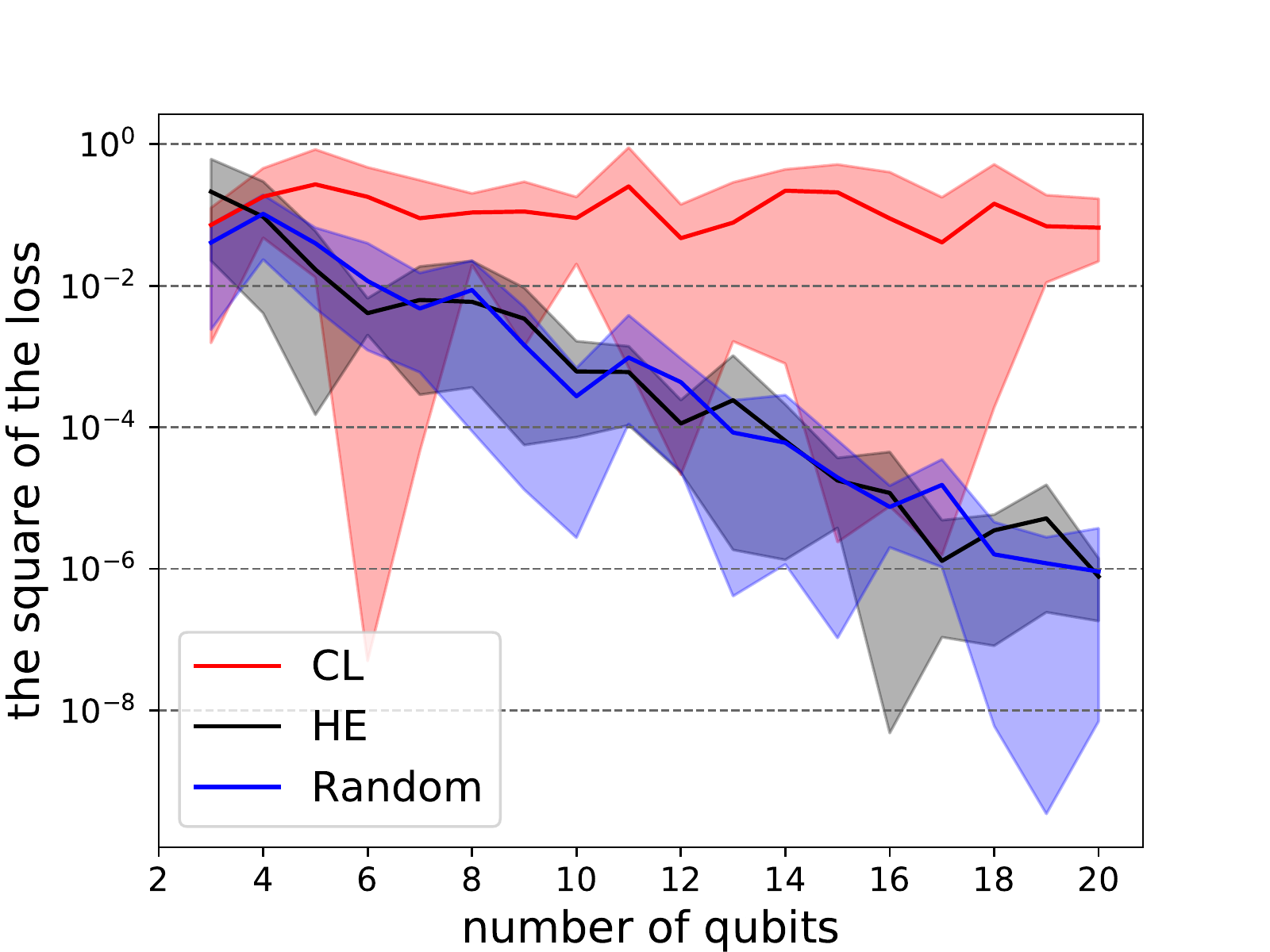}  
  \label{tqnn_fig_toy_loss}
}
 \subfigure[]{
  \centering
  \includegraphics[width=.4\linewidth]{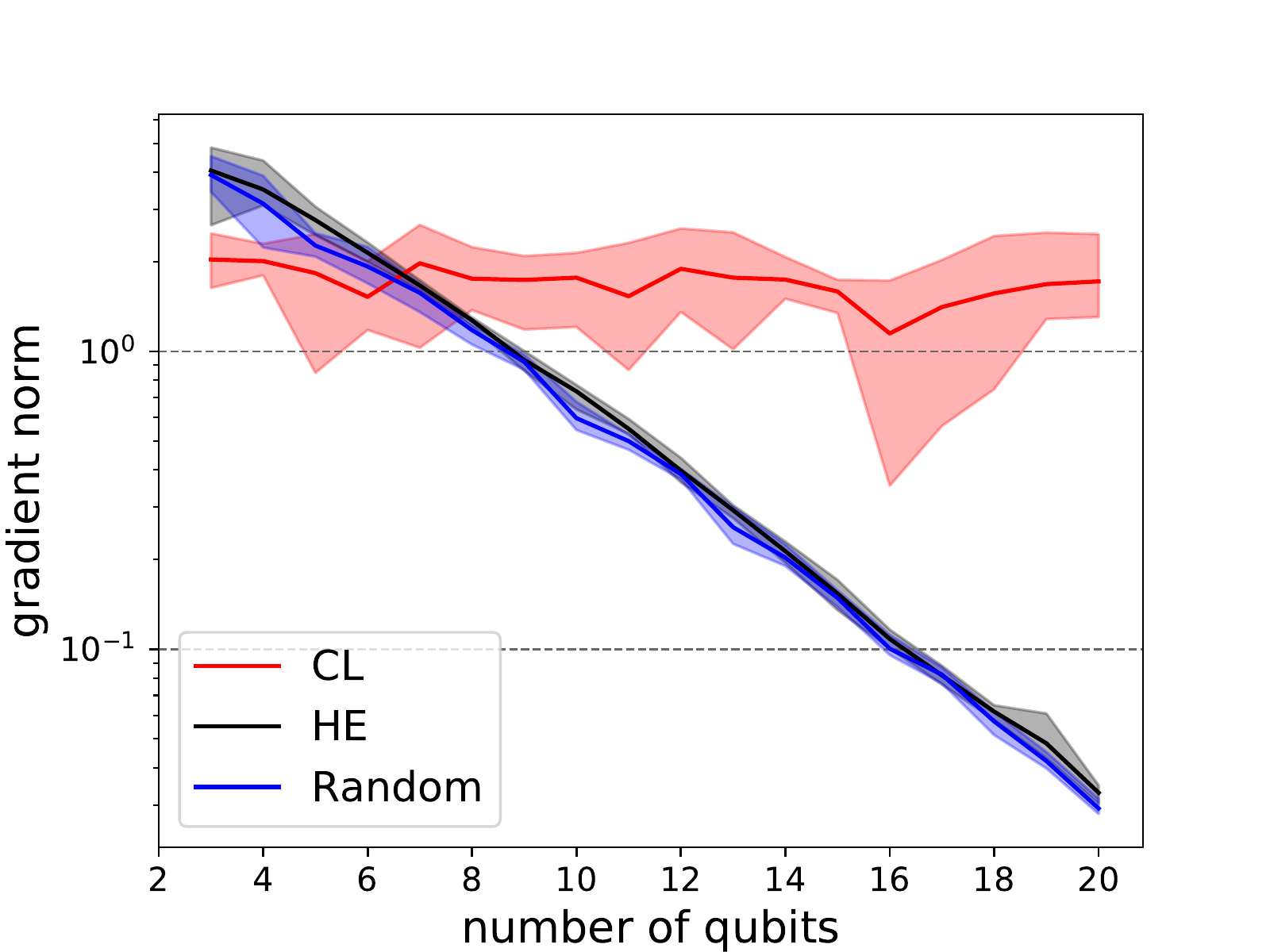}  
  \label{tqnn_fig_toy_grad}
}
  \subfigure[]{
  \centering
  \includegraphics[width=.4\linewidth]{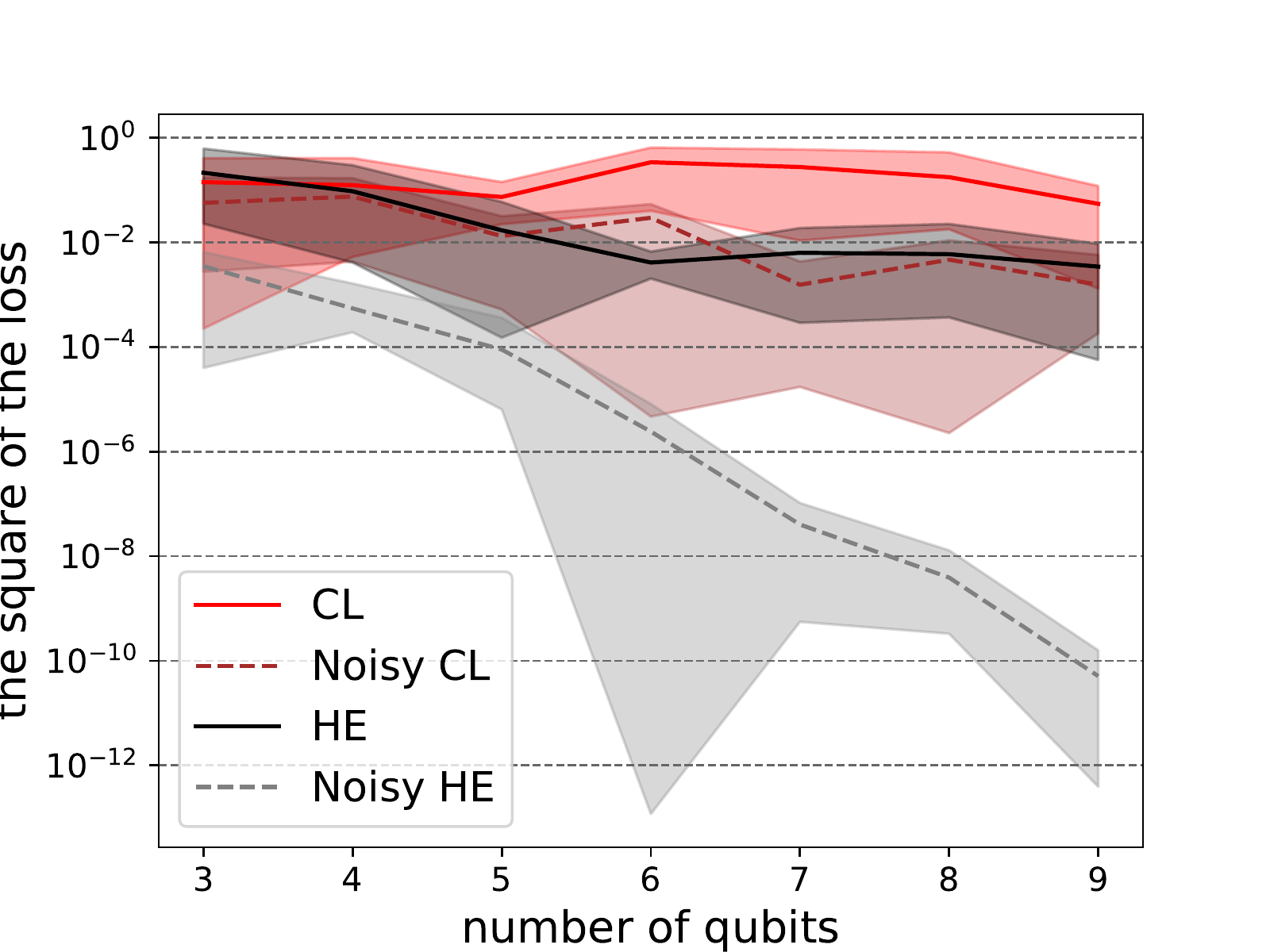}  
  \label{tqnn_fig_toy_noisy_loss}
}
  \subfigure[]{
  \centering
  \includegraphics[width=.4\linewidth]{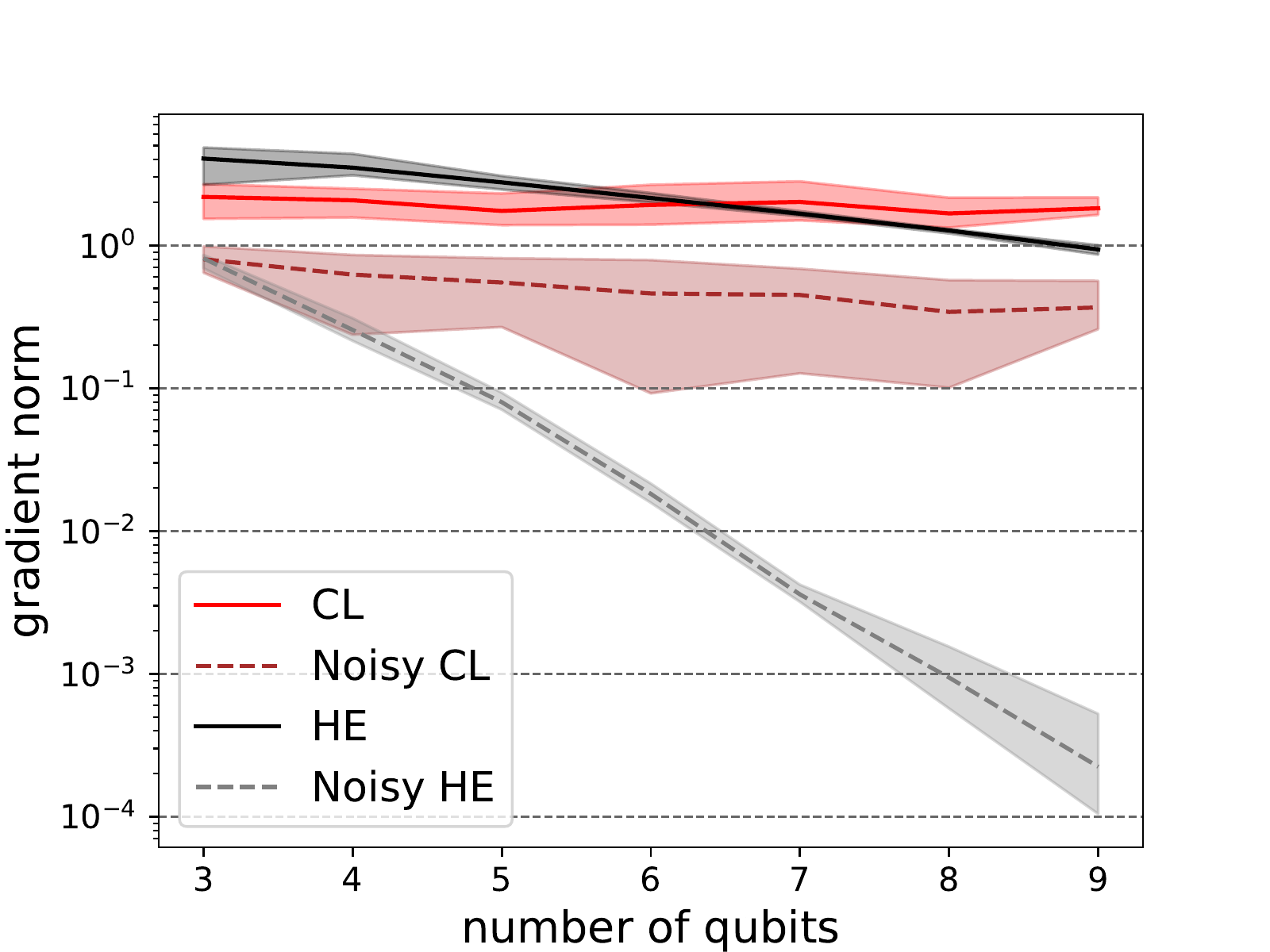}  
  \label{tqnn_fig_toy_noisy_grad}
}
 \subfigure[]{
  \centering
  \includegraphics[width=.4\linewidth]{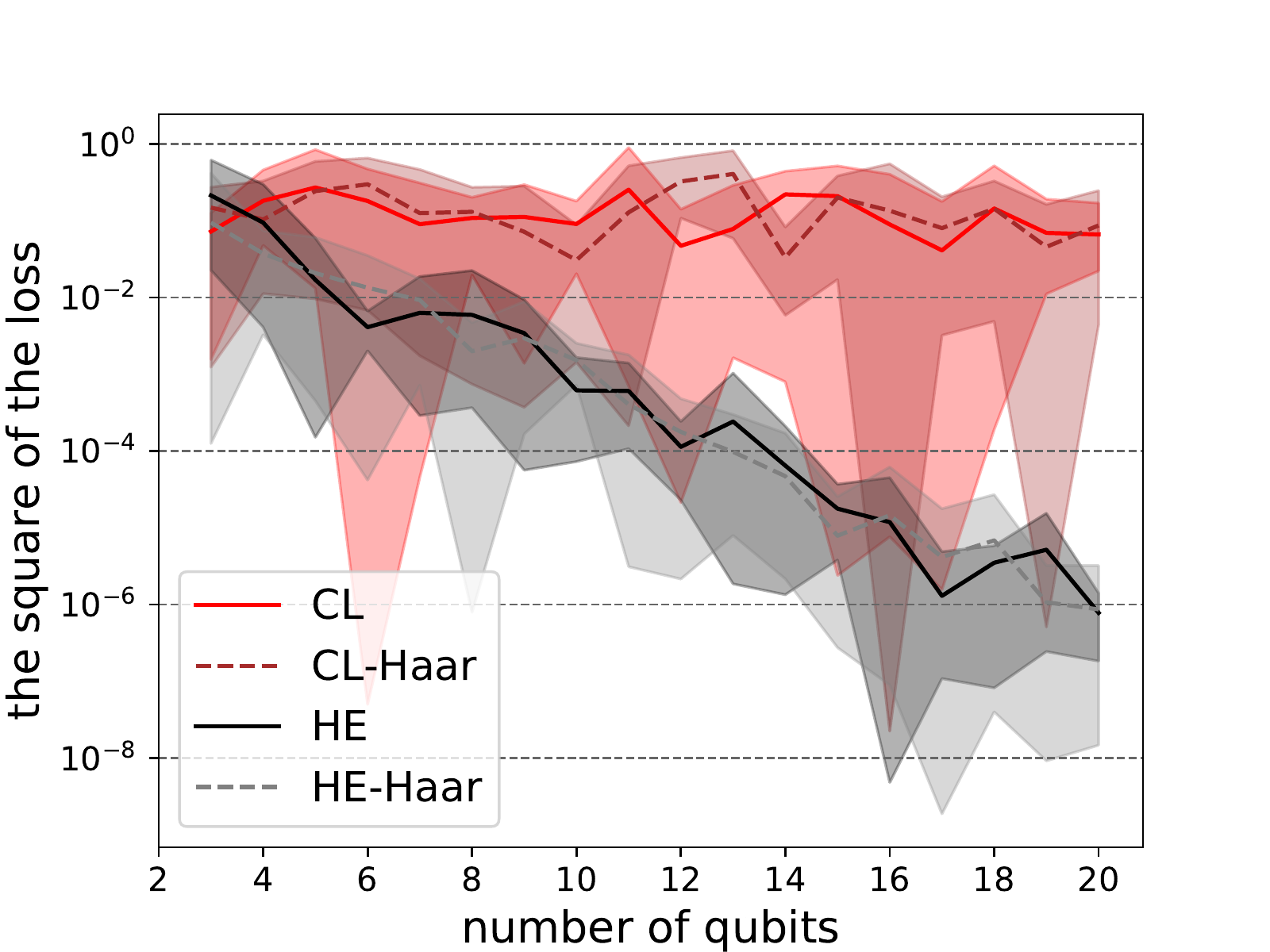}  
  \label{tqnn_fig_toy_haar_loss}
}
 \subfigure[]{
  \centering
  \includegraphics[width=.4\linewidth]{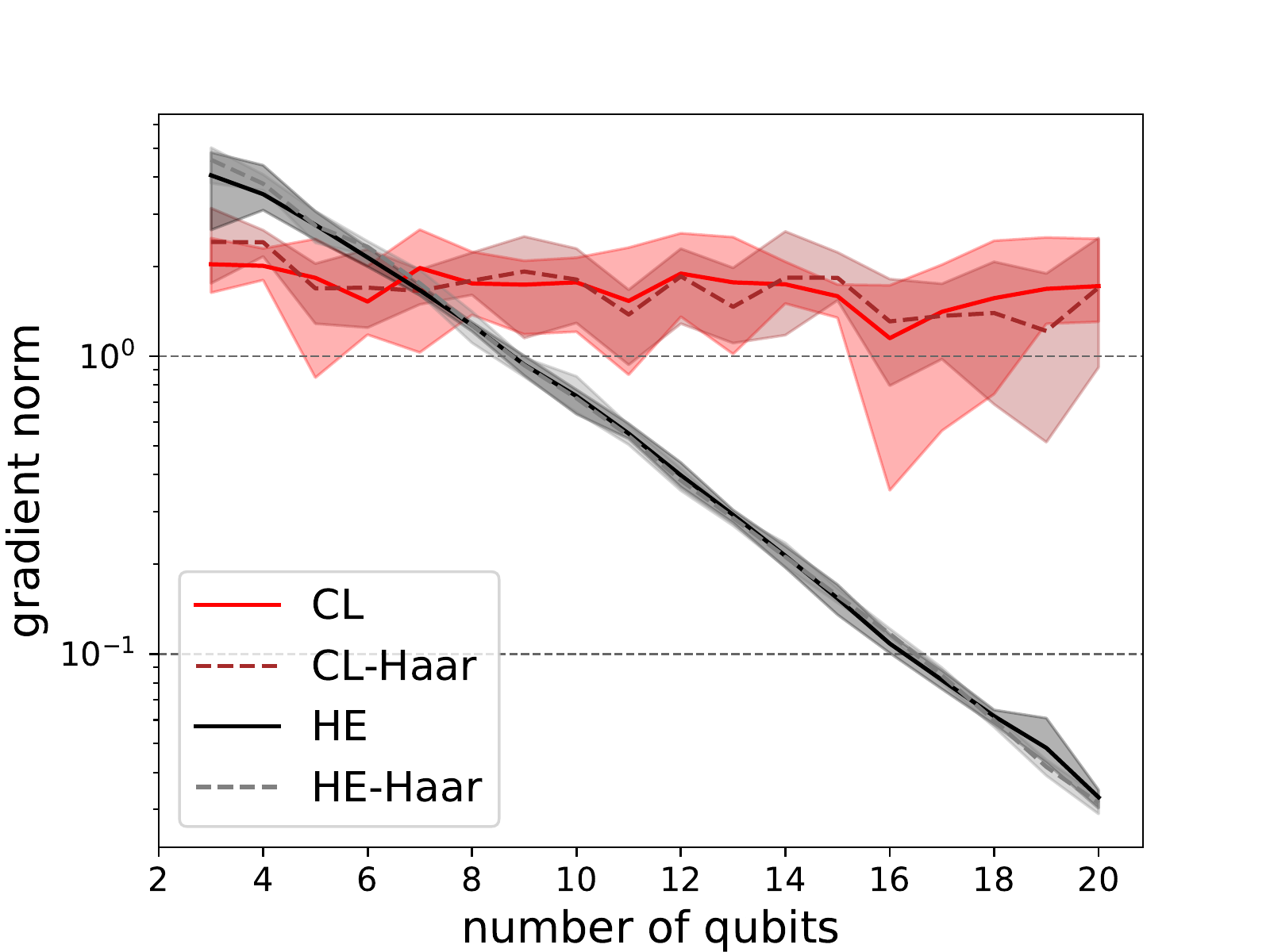}  
  \label{tqnn_fig_toy_haar_grad}
}
\caption{{
Numerical results of the loss value and the gradient information of the toy model~(\ref{tqnn_toy_eq}). Figures~\ref{tqnn_fig_toy_loss} and \ref{tqnn_fig_toy_grad} show the squared loss and the gradient norm with increasing qubits $N \in \{3,4,\cdots,20\}$, respectively. Figures~\ref{tqnn_fig_toy_noisy_loss} and \ref{tqnn_fig_toy_noisy_grad} show the squared loss and the gradient norm corresponding to noiseless and noisy CL-QNNs and HE-QNNs with increasing qubits $N \in \{3,4,\cdots,9\}$, respectively. The noisy channel follow the Figure~\ref{tqnn_circuit_noise}. Figures~\ref{tqnn_fig_toy_haar_loss} and \ref{tqnn_fig_toy_haar_grad} show the squared loss and the gradient norm corresponding to CL-QNNs and HE-QNNs initialized with the uniform parameter distribution and the Haar distribution, respectively.
Each line denotes the average of $5$ rounds of independent simulations.}}
\end{figure*}

{
To begin with, we study the toy model,
\begin{equation}\label{tqnn_toy_eq}
f(\boldsymbol{\theta}) = \Tr[O V(\boldsymbol{\theta}) (|0\>\<0|)^{\otimes N} V(\boldsymbol{\theta})^{\dag}],
\end{equation}
where the observable $O=\sigma_3\otimes \sigma_0 \otimes \cdots \otimes \sigma_0$. We compare the performance of three different ansatzes, the CL-QNN, the hardware-efficient QNN (HE-QNN) in Figure~\ref{tqnn_he_circuit}, and the randomly structured QNN. For CL-QNNs, we set the number of CL blocks $L=2$ and perform CZ gates on all neighboring qubits $(1,2),(2,3),\cdots,(N-1,N),(N,1)$. The number $S=1$ by considering the observable in Eq.~(\ref{tqnn_toy_eq}). Unitaries on the remaining $N-S$ qubits (the light blue block in Figure~\ref{tqnn_main_circuit}) are chosen to be $L_{HE}=5$ layers of the hardware-efficient ansatz in Figure~\ref{tqnn_he_circuit}. Thus, the depth of the parameterized part of CL-QNNs is $D_{CL} = 3 L_{HE}L = 30$. To make a fair comparison, the number of single-qubit gates, CZ gates, and CNOT gates in Random-QNNs are set to be the same as that in the CL-QNNs. The parameterized circuit layer of HE-QNNs is set to be the same as that in CL-QNNs, which is constructed by repeating the hardware-efficient ansatz for $L_{HE}=10$ layers. Since the depth of one HE layer is $\mathcal{O}(N)$, the depth of CL-QNNs and HE-QNNs would grow linearly with increasing qubits.
Parameters in different QNNs are initialized independently from a uniform distribution in $[0,2\pi]$ if no additional requirements are made. }

{
Numerical results of the toy model are shown in Figures~\ref{tqnn_fig_toy_loss}-\ref{tqnn_fig_toy_haar_grad}. First, we record the squared loss and the gradient norm of Eq.~(\ref{tqnn_toy_eq}) for noiseless CL-QNNs, HE-QNNs, and Random-QNNs in Figures~\ref{tqnn_fig_toy_loss} and \ref{tqnn_fig_toy_grad}. For noiseless HE-QNNs and Random-QNNs, both the squared loss and the gradient norm decay exponentially with increasing qubit numbers $N \in \{3,4,\cdots,20\}$, while for noiseless CL-QNNs, these values remain the same magnitude. Thus, we have verified  Theorems~\ref{tqnn_appendix_lemma_cost} and \ref{tqnn_appendix_theorem}. }

\begin{figure*}[t]
 \subfigure[]{
  \centering
  \includegraphics[width=.4\linewidth]{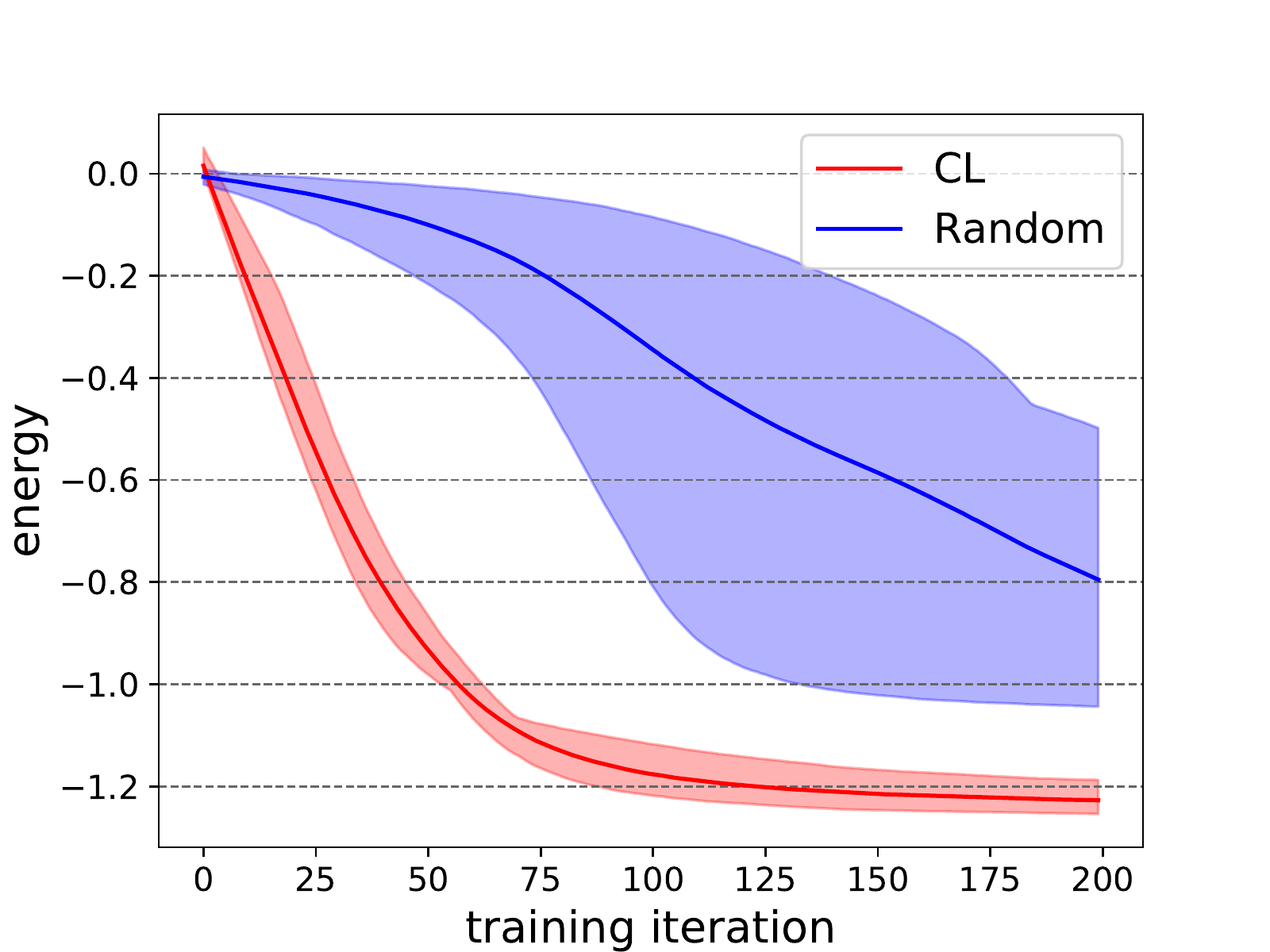}  
  \label{tqnn_fig_ising_sgd_loss}
}
 \subfigure[]{
  \centering
  \includegraphics[width=.4\linewidth]{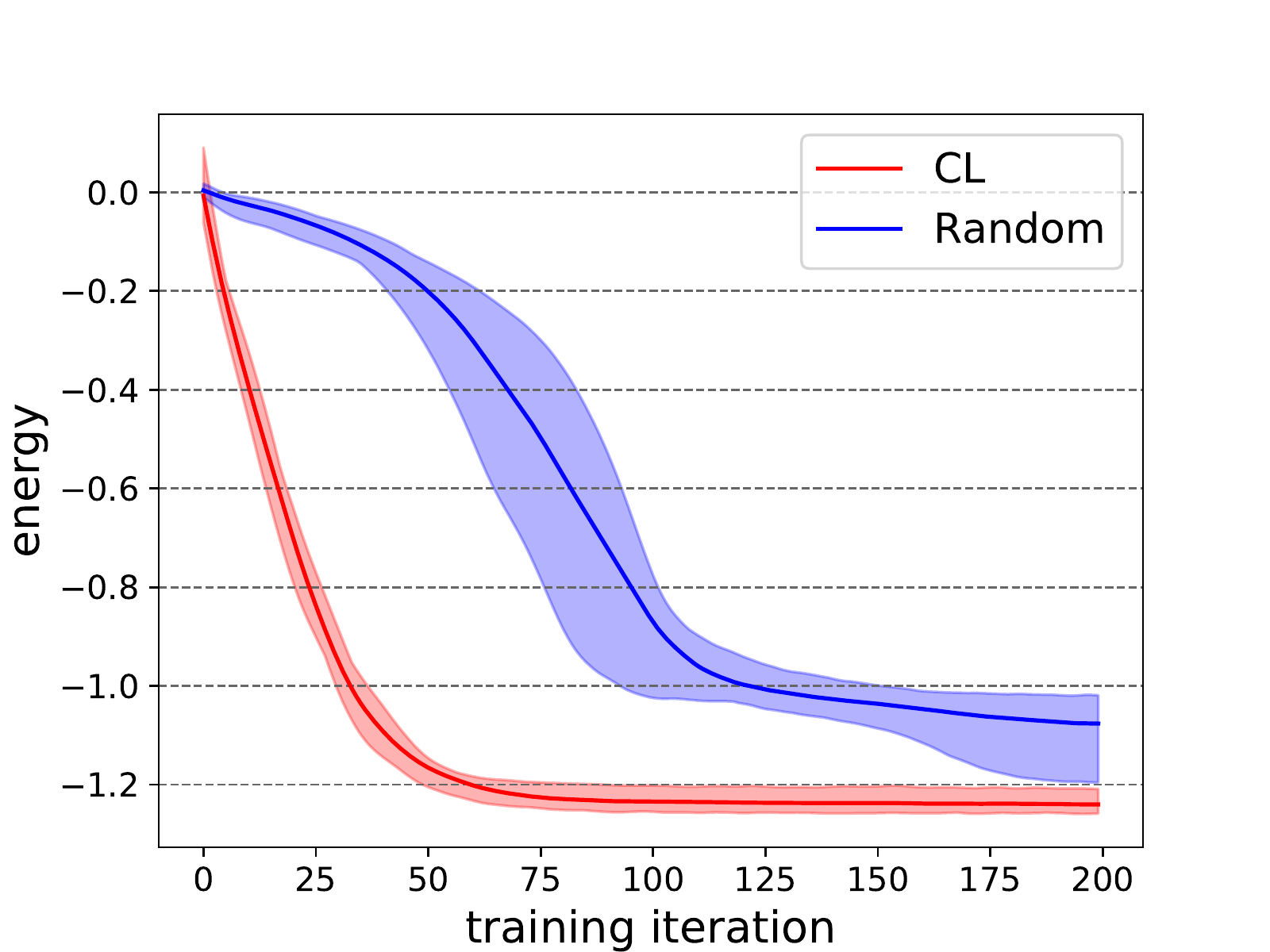}  
  \label{tqnn_fig_ising_sadam_loss}
}
  \subfigure[]{
  \centering
  \includegraphics[width=.4\linewidth]{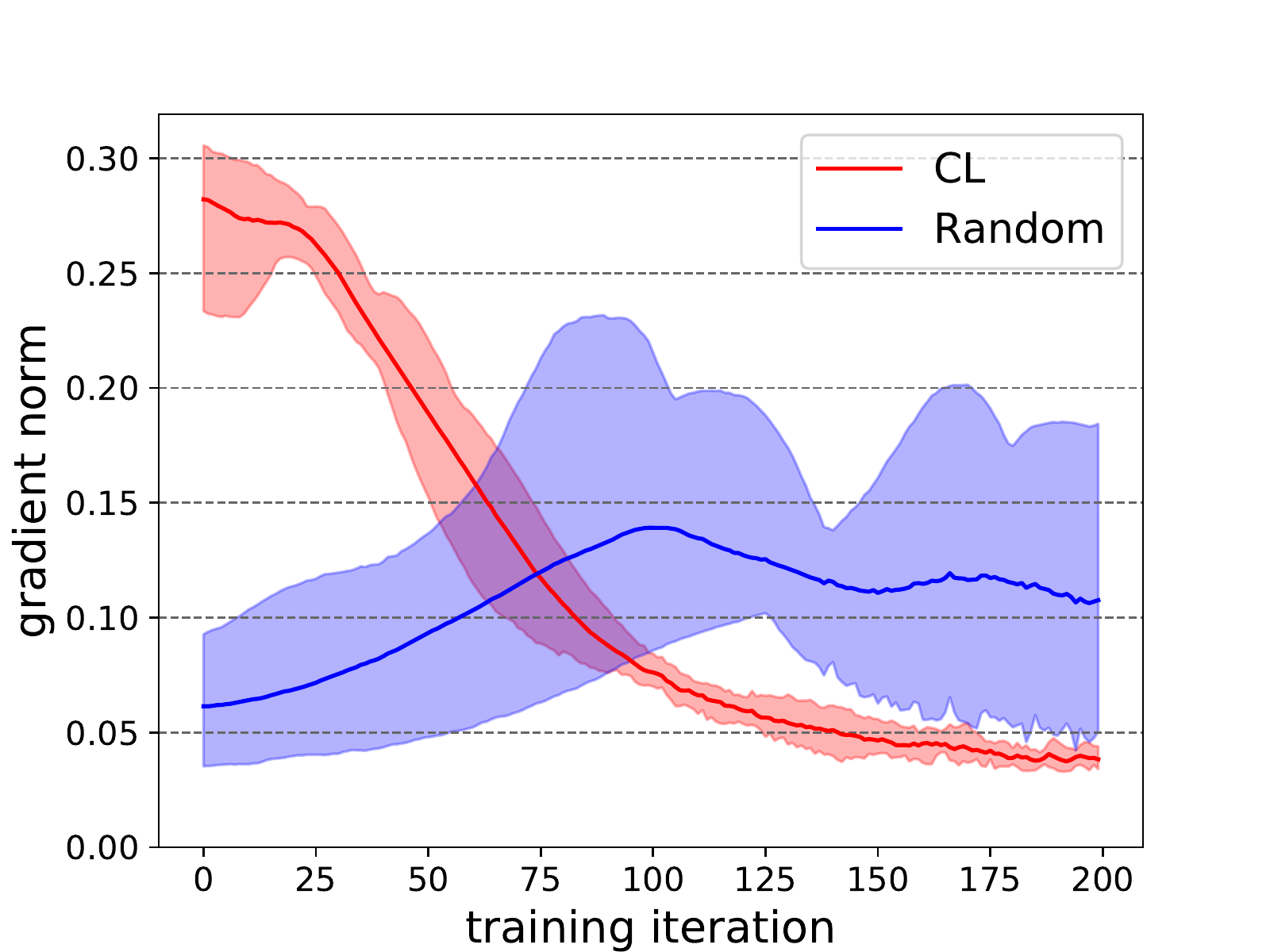}  
  \label{tqnn_fig_ising_sgd_gradnorm}
}
  \subfigure[]{
  \centering
  \includegraphics[width=.4\linewidth]{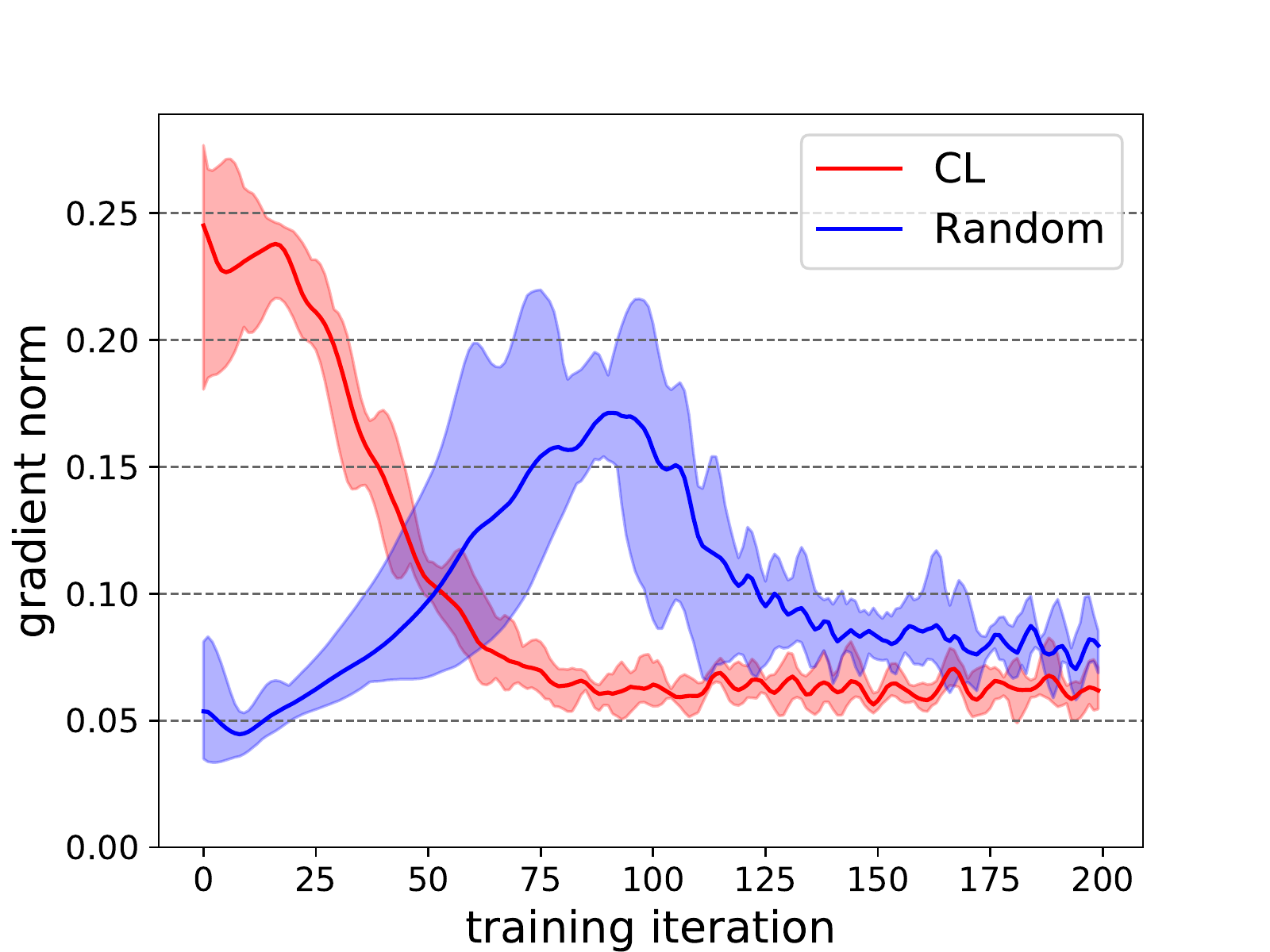}  
  \label{tqnn_fig_ising_sadam_gradnorm}
}
\caption{Numerical results of QNN finding the ground state energy of the Ising model. Figures~\ref{tqnn_fig_ising_sgd_loss} and \ref{tqnn_fig_ising_sadam_loss} show the loss corresponding to the Hamiltonian of the Ising model (\ref{tqnn_ising_eq}) with stochastic gradient descent and stochastic Adam optimizers, respectively. Figures~\ref{tqnn_fig_ising_sgd_gradnorm} and \ref{tqnn_fig_ising_sadam_gradnorm} show the $\ell_2$-norm of the corresponding gradient during the training.
Red and blue lines denote the average of $5$ rounds for the CL-QNN and the Random-QNN, respectively.}
\end{figure*}

{
Next, we consider the performance of CL-QNNs under noisy settings. Specifically, we consider the depolarizing noisy channel illustrated in Figure~\ref{tqnn_circuit_noise}. The noisy channel on the single-qubit state acts as
\begin{equation}\label{tqnn_toy_noisechannel_eq}
\mathcal{N}_q(\rho) = q \rho + \frac{1-q}{2} I,
\end{equation}
where we set $q=0.99$ in simulations.
We compare the squared loss and the gradient norm of CL-QNNs and HE-QNNs under noiseless and noisy settings for increasing qubit numbers $N \in \{3,4,\cdots,9\}$ in Figures~\ref{tqnn_fig_toy_noisy_loss} and \ref{tqnn_fig_toy_noisy_grad}, respectively. The squared loss of both CL-QNNs and HE-QNNs decays when the noise is introduced. Specifically, the squared loss of CL-QNNs drops by around one order of magnitude when the qubit number increases from $3$ to $9$, while the squared loss of HE-QNNs drops by around four orders of magnitude. We observe similar results about the gradient norm in Figure~\ref{tqnn_fig_toy_noisy_grad}. Thus, the trainability of CL-QNNs is less influenced by noise-induced barren plateaus~\cite{wang2021noise}, compared to that of HE-QNNs. } 

{
Finally, we compare the performance of different parameter initializations, i.e., the uniform distribution in $[0,2\pi]$ and the Haar distribution in Figures~\ref{tqnn_fig_toy_haar_loss} and \ref{tqnn_fig_toy_haar_grad}, which show the squared loss and the gradient norm of CL-QNNs and HE-QNNs with qubits $N \in \{3,4,\cdots,20\}$, respectively. In general, parameterized circuits with restricted structures could not generate the Haar distribution in the whole unitary space $\mathcal{U}(2^N)$, so we consider the local Haar distribution on each qubit, which could be generated as $R_X(\theta_3) R_Y(\theta_2) R_X(\theta_1)$ by ignoring the global phase. As shown in Figures~\ref{tqnn_fig_toy_haar_loss} and \ref{tqnn_fig_toy_haar_grad}, for both CL-QNNs and HE-QNNs, the uniform distribution and the Haar distribution show similar performances. }

\begin{figure*}[t]
 \subfigure[]{
  \centering
  \includegraphics[width=.4\linewidth]{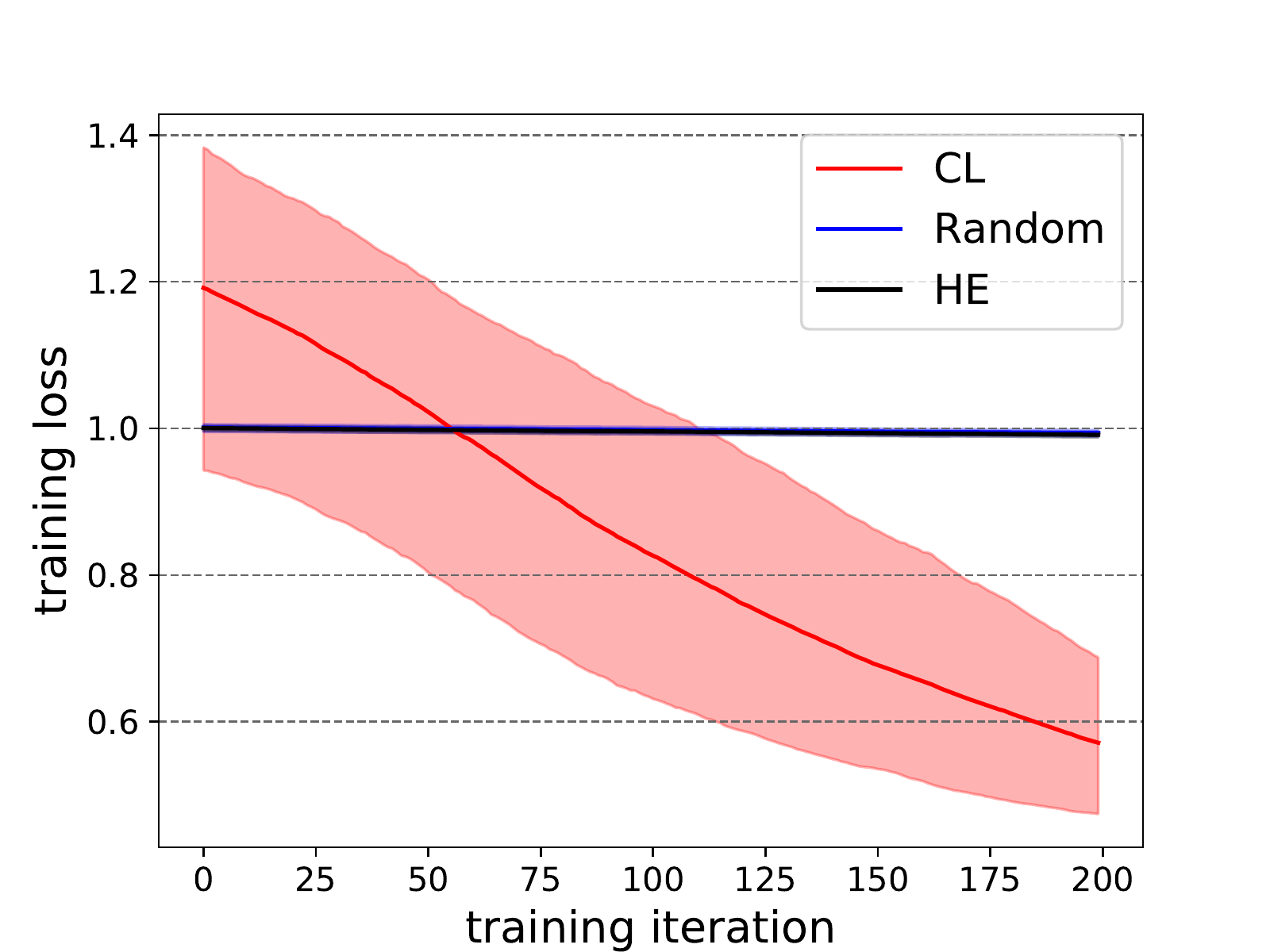}  
  \label{tqnn_fig_qml_sgd_loss}
}
 \subfigure[]{
  \centering
  \includegraphics[width=.4\linewidth]{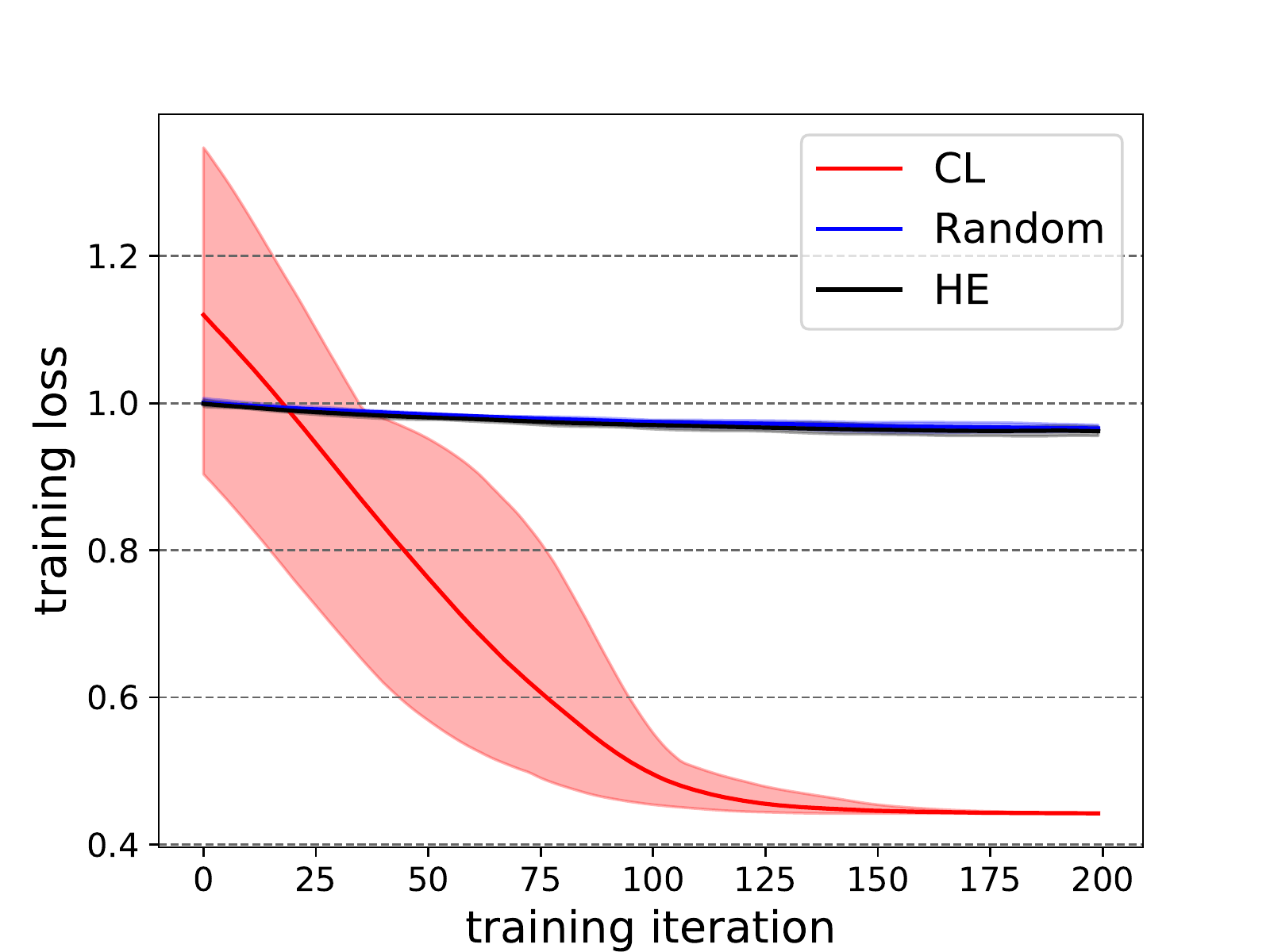}  
  \label{tqnn_fig_qml_sadam_loss}
}
 \subfigure[]{
  \centering
  \includegraphics[width=.4\linewidth]{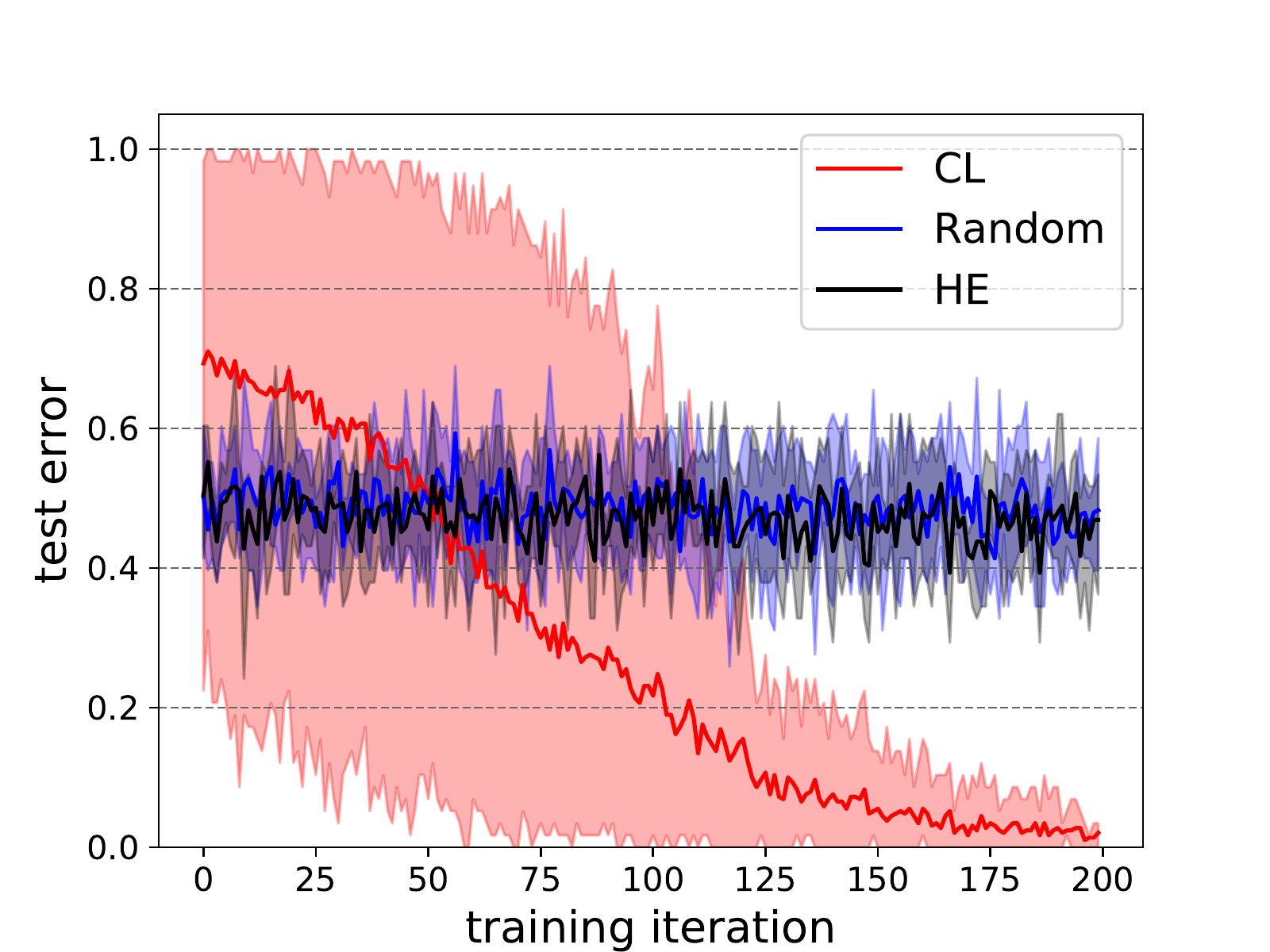}  
  \label{tqnn_fig_qml_sgd_error}
}
 \subfigure[]{
  \centering
  \includegraphics[width=.4\linewidth]{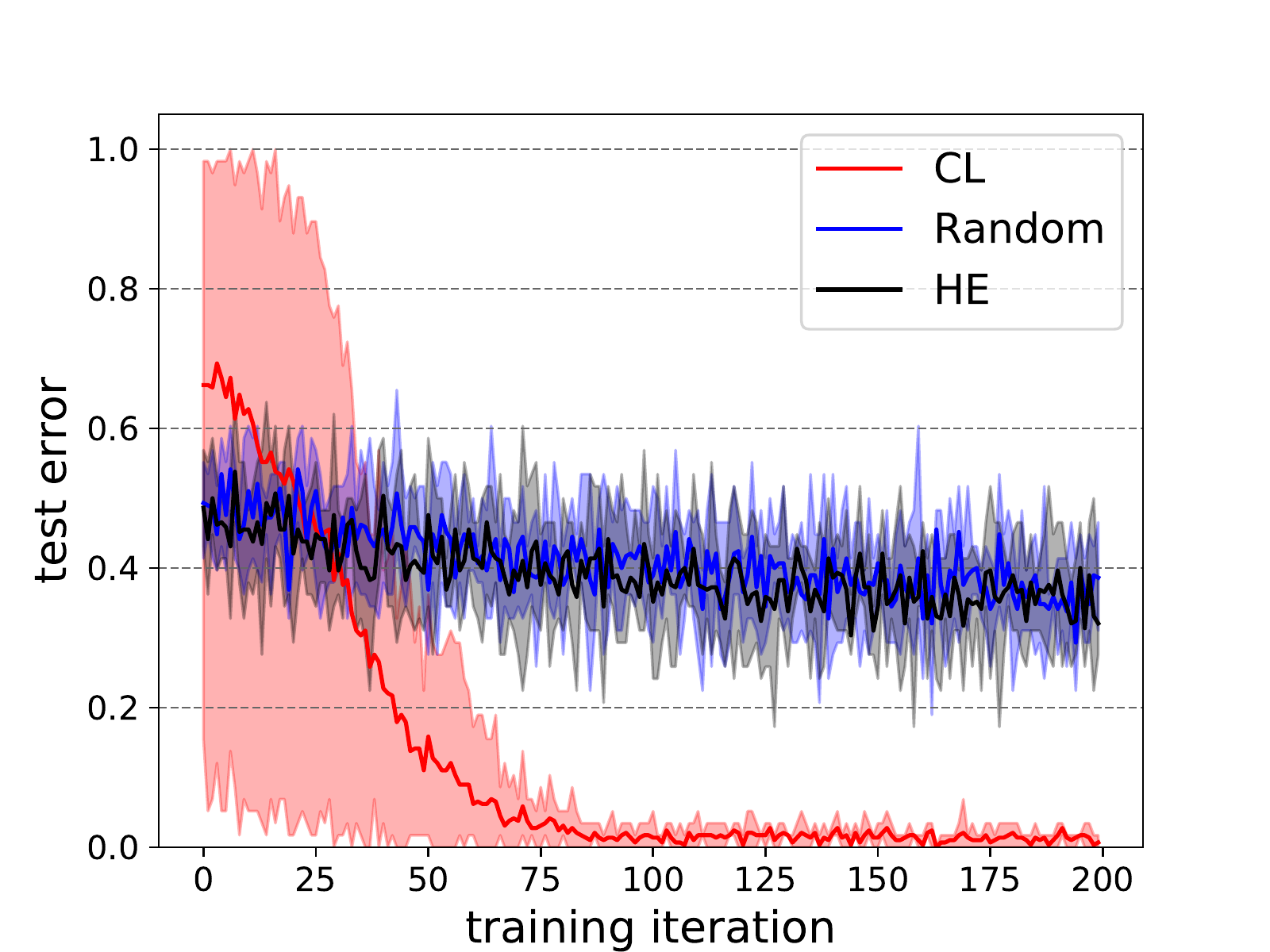}  
  \label{tqnn_fig_qml_sadam_error}
}
  \subfigure[]{
  \centering
  \includegraphics[width=.4\linewidth]{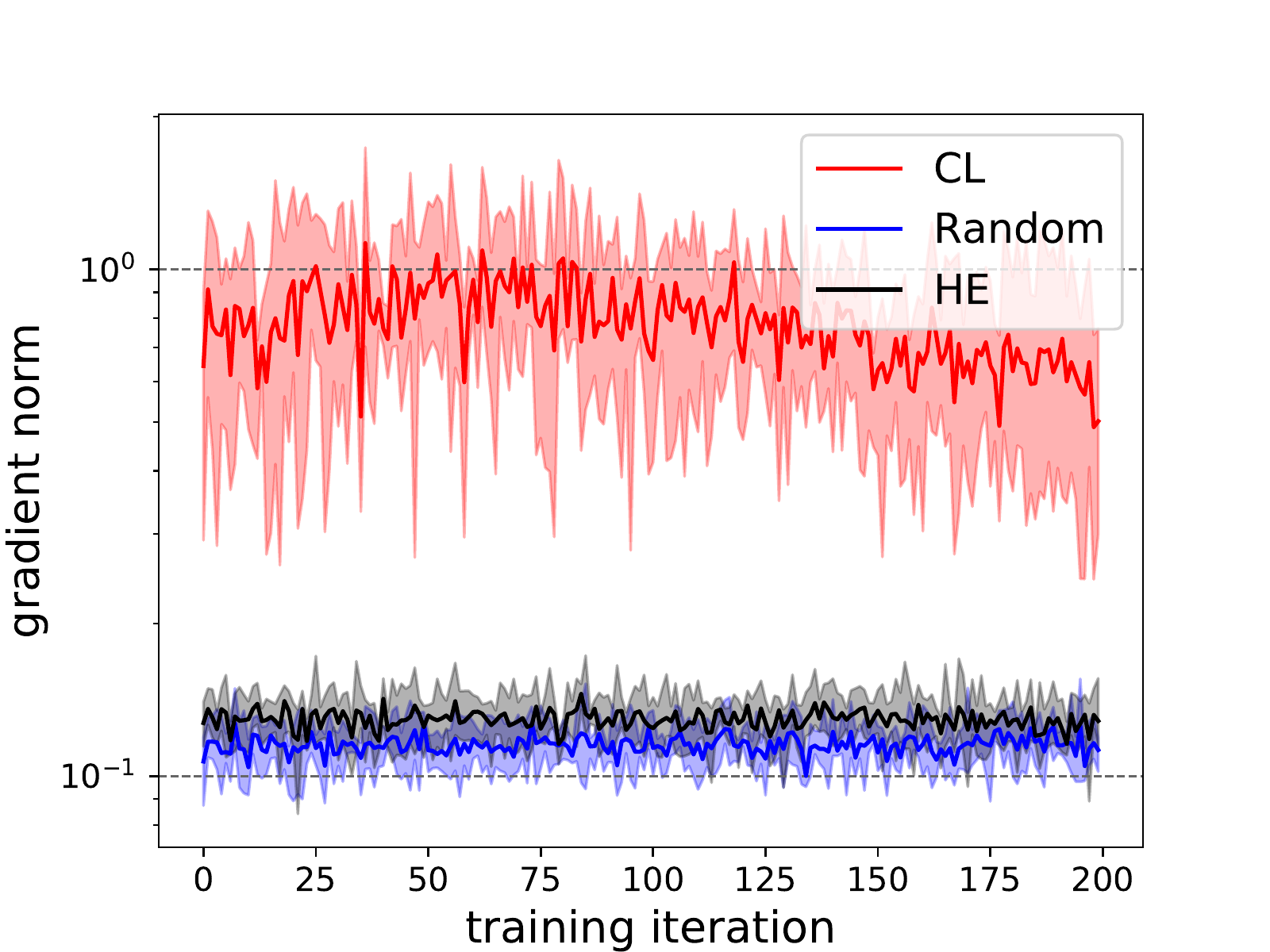}  
  \label{tqnn_fig_qml_sgd_gradnorm}
}
  \subfigure[]{
  \centering
  \includegraphics[width=.4\linewidth]{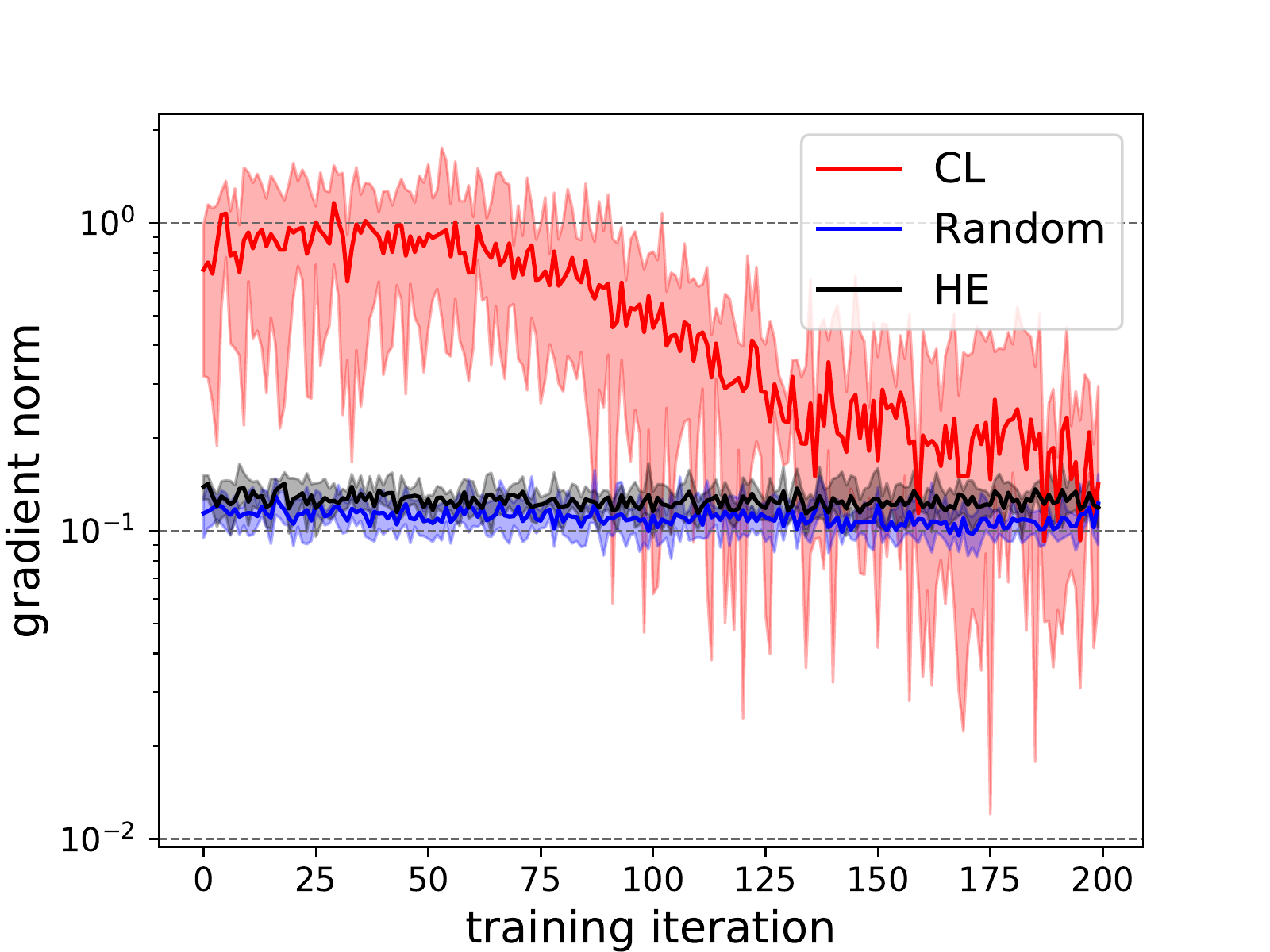}  
  \label{tqnn_fig_qml_sadam_gradnorm}
}
\caption{Numerical results of the binary classification on the wine dataset~\cite{Dua2019}. Figures~\ref{tqnn_fig_qml_sgd_loss} and \ref{tqnn_fig_qml_sadam_loss} show the loss (\ref{tqnn_qml_eq}) of the training set with stochastic gradient descent and stochastic Adam optimizers. Figures~\ref{tqnn_fig_qml_sgd_error} and \ref{tqnn_fig_qml_sadam_error} show the error (\ref{tqnn_qml_error_eq}) of the test set with stochastic gradient descent and stochastic Adam optimizers. Figures~\ref{tqnn_fig_qml_sgd_gradnorm} and \ref{tqnn_fig_qml_sadam_gradnorm} show the $\ell_2$-norm of the corresponding gradient during the training.
Red, blue, and black lines denote the average of $5$ rounds for CL-QNN, the Random-QNN, and the hardware-efficient QNN (HE-QNN), respectively.}
\label{tqnn_fig_qml}
\end{figure*}

\subsection{Ising model}
\label{tqnn_exper_ising}

In the second task, we aim to find the ground state energy of the transverse-field Ising model~\cite{pfeuty1970one} around the critical point with periodic boundary conditions, i.e.
\begin{equation}\label{tqnn_ising_eq}
H = -\frac{1}{N} \sum_{i=1}^{N} 	Z_i Z_{i+1} - \frac{1}{N} \sum_{i=1}^{N} 	X_i,
\end{equation}
where we employ the periodic boundary condition $Z_{N+1}=Z_{1}$. We adopt the loss function defined by Eq.~(\ref{tqnn_intro_loss_function}) with the corresponding observable in Eq.~(\ref{tqnn_ising_eq}). We compare the performance of CL-QNNs with that of randomly structured QNNs. Both QNNs contain $N=16$ qubits and begin from the initial state $(|0\>\<0|)^{\otimes N}$. We set the number of CL blocks $L=6$ and perform CZ gates on all neighboring qubits $(1,2),(2,3),\cdots,(N-1,N),(N,1)$. Unitaries on the remaining $N-S$ qubits in Figure~\ref{tqnn_main_circuit} are employed as the tensor product of $R_X R_Y R_X$ gate sequences. Overall, we generate a CL-QNN with $270$ single-qubit gates and $75$ CZ gates. To make a fair comparison, the number of single-qubit gates and CZ gates in the Random-QNN are set to be the same as that in the CL-QNN. Initial parameters in both QNNs are sampled independently with a uniform distribution in $[0,2\pi]$. We employ two different optimizers, i.e., stochastic gradient descent (SGD)~\cite{Bottou2012} and stochastic Adam {(SGD with adaptive momentum)}~\cite{kingma2014adam}, where the randomness is introduced by finite measurements. In the experiment, we set the number of measurements to be $100$ for each Pauli observable, and the number of training iterations is $200$. The learning rate is $0.15$ for SGD, while for stochastic Adam, the learning rate is $0.01$.

The numerical results of the Ising model task are shown in Figures~\ref{tqnn_fig_ising_sgd_loss}-\ref{tqnn_fig_ising_sadam_gradnorm}. The loss during the training iteration is illustrated in Figures~\ref{tqnn_fig_ising_sgd_loss} and \ref{tqnn_fig_ising_sadam_loss} for different QNNs. For the SGD case, we notice the convergence of the loss corresponding to the CL-QNN after $150$ iterations, while the loss of Random-QNN does not converge and the final value is much higher than that of the CL-QNN. For the stochastic Adam case, the loss of the CL-QNN converges after $75$ iterations, while the loss of Random-QNN does not converge with a much higher final value. We record the $\ell_2$-norm of the gradient during the training in Figures~\ref{tqnn_fig_ising_sgd_gradnorm} and \ref{tqnn_fig_ising_sadam_gradnorm}. The gradient norm for the CL-QNN is larger than $0.2$ at the initial point and converges during training. On the other hand, the gradient norm of the Random-QNN is much smaller than that of the CL-QNN at the initial point and then exceeds the CL-QNN during training. The gradient norm of the Random-QNN after training indicates that the loss of the Random-QNN converges to its stationary point slower than that of the CL-QNN.

\subsection{Binary classification}
\label{tqnn_exper_qml} 

In the third task, we aim to classify two kinds of wines based on their features. The wine dataset \cite{Dua2019} consists of 138 instances derived from three different cultivars. Each of the instances has 13 features that measure different chemical components. In this task, we choose two kinds of wines and divide the original dataset into the training and test set. Both datasets contain 58 samples, and the number of samples from two different classes are the same. Specifically, we denote $A=\{(\boldsymbol{x}^{(i)},y^{(i)})\}$ as the set of input data, where  $y^{(i)} = \pm 1$ for different labels. Qubit embedding~\cite{lloyd2020quantum} is employed to encode the information of datasets into quantum states, i.e.,
\begin{equation*}
\rho_{\text{in}}(\boldsymbol{x}) = \bigotimes\limits_{j=1}^{\text{dim}(\boldsymbol{x})} R_Y(x_j) |0\>\<0| R_Y(x_j)^{\dag},
\end{equation*}
where $\text{dim}(\boldsymbol{x})=13$ is the dimension of $\boldsymbol{x}$. We optimize the $\ell_1$-norm loss defined by
\begin{equation}\label{tqnn_qml_eq}
f_{\text{QML}}(\boldsymbol{\theta}) = \frac{1}{|A|} \sum_{i=1}^{|A|} \left| \Tr[O V(\boldsymbol{\theta}) \rho_{\text{in}}(\boldsymbol{x}^{(i)}) V(\boldsymbol{\theta})^\dag] - y^{(i)} \right|,
\end{equation}
where $O=\sigma_3 \otimes \sigma_0 \cdots \otimes \sigma_0$, and $|A|$ is the size of the dataset. Since the aim of the binary classification is to arrange the right label for the unseen data, we will record the classification error 
\begin{equation}\label{tqnn_qml_error_eq}
e(\boldsymbol{\theta}) = \frac{1}{|A|} \sum_{i=1}^{|A|} \text{sign}\left( \Tr[O V(\boldsymbol{\theta}) \rho_{\text{in}}(\boldsymbol{x}^{(i)}) V(\boldsymbol{\theta})^\dag] - y^{(i)} \right)
\end{equation}
on the test dataset.

We compare the performance of CL-QNNs with that of hardware-efficient QNNs (HE-QNNs) in Figure~\ref{tqnn_he_circuit} and randomly structured QNNs. The qubit number $N=13$ suffices to encode all the features of the wine dataset. For CL-QNNs, we set the number of CL blocks {$L=2$} and perform CZ gates on all neighboring qubits $(1,2),(2,3),\cdots,(N-1,N),(N,1)$. The number $S=1$ by considering the observable in Eq.~(\ref{tqnn_qml_eq}). Unitaries on the remaining $N-S$ qubits (the light blue block in Figure~\ref{tqnn_main_circuit}) are chosen to be the hardware-efficient ansatz in Figure~\ref{tqnn_he_circuit} with the depth $L_{HE}=5$. Thus, the depth of the parameterized part of the CL-QNN is {$D_{CL} = 3 L_{HE} L = 30$}. To make a fair comparison, the number of single-qubit gates, CZ gates, and CNOT gates in Random-QNNs are set to be the same as that in the CL-QNNs, and the parameterized circuit depth of the HE-QNNs is set to be the same as that in the CL-QNN by using {$L_{HE}=10$} HE layers. Parameters in different QNNs are initialized independently from a uniform distribution in $[0,2\pi]$. We employ SGD and stochastic Adam optimizers, where the randomness is introduced by finite measurements and the batch training strategy. We set the number of measurements to be $100$ and the batch size to be $8$. The number of training iterations is $200$. The learning rates of SGD and stochastic Adam optimizers are $0.01$.

The numerical results of the binary classification task using different QNNs are shown in Figures~\ref{tqnn_fig_qml_sgd_loss}-\ref{tqnn_fig_qml_sadam_gradnorm}. {For the SGD case, both the training loss and the test error corresponding to the CL-QNN decrease significantly compared with that of the Random-QNN and the HE-QNN. For the stochastic Adam case, the training loss and the test error of the CL-QNN converge after $125$ iterations, while the loss of Random-QNN and HE-QNN do not show clear convergence with a much higher final value. We record the $\ell_2$-norm of the gradient during the training in Figures~\ref{tqnn_fig_qml_sgd_gradnorm} and \ref{tqnn_fig_qml_sadam_gradnorm}. For the SGD case, the gradient norm of the CL-QNN is larger than $0.5$ during the training. For the stochastic Adam case, the gradient norm for the CL-QNN is larger than $0.5$ at the initial point and converges during training. On the other hand, the gradient norm of the Random-QNN and the HE-QNN is much smaller than that of the CL-QNN at the initial point before the latter converges.}

In conclusion, CL-QNNs of moderate size showed much better trainability for finding the ground state of the Ising model and the binary classification when compared with Random-QNNs and HE-QNNs. These numerical experiments accord with the theorems obtained in the previous section.

\section{Conclusion} \label{tqnn_conclusions}

In this work, we analyze the vanishing gradient problem in quantum neural networks. We prove that the gradient norm of $N$-qubit quantum neural networks with a controlled-layer structure is lower bounded by $\Omega(8^{-LS})$, where $L$ is the number of CL blocks in the circuit and $S$ is the number of qubits corresponding to the quantum observable. We remark that the circuit depth in each CL block could be arbitrary. Thus, the bound guarantees the trainability of deep CL-QNNs. 

Our theoretical framework does not require the unitary 2-design assumption as in previous works; hence, it will be more applicable to real-world quantum neural networks for near-term quantum computers. 
When compared with randomly structured QNNs and hardware-efficient QNNs, which are known to suffer from the barren plateau problem, CL-QNNs show better trainability for finding the ground state of the Ising model and the binary classification task. We hope that the paper could inspire future works on the trainability of deep QNNs with different architectures and other variational quantum algorithms.

\bibliography{reference}
\bibliographystyle{unsrt}

\newpage


\appendix

\section{Technical Lemmas}
\label{tqnn_app_technical_lemmas}
In this section we provide some technical lemmas.

Consider the objective function of the QNN defined as 
\begin{equation*}
f(\boldsymbol{\theta}) = \<O\> = {\rm Tr} [O \cdot V(\boldsymbol{\theta}) \rho_{\text{in}} V(\boldsymbol{\theta})^{\dag}],	
\end{equation*}
where $\boldsymbol{\theta}$ encodes all parameters which participate the circuit as $e^{-i\theta_j \sigma_k}, k \in {1,2,3}$, $\rho_{\text{in}}$ denotes the input state and $O$ is an arbitrary quantum observable. 
\begin{lemma}\label{tqnn_params_shifting}
The partial derivative of the function $f(\boldsymbol{\theta})$ with respect to the parameter $\theta_j$ could be calculated by
\begin{multline}\label{tqnn_appendix_lemma_shifting_rule}
\frac{\partial f}{\partial \theta_j}  = {\rm Tr} [O \cdot V(\boldsymbol{\theta}_+) \rho_{\rm {in}} V(\boldsymbol{\theta}_+)^{\dag}] \\ - {\rm Tr} [O \cdot V(\boldsymbol{\theta}_-) \rho_{\rm {in}} V(\boldsymbol{\theta}_-)^{\dag}] ,
\end{multline}
where $\boldsymbol{\theta_+} \equiv \boldsymbol{\theta} + \frac{\pi}{4}\boldsymbol{e}_j$ and $\boldsymbol{\theta_-} \equiv \boldsymbol{\theta} - \frac{\pi}{4}\boldsymbol{e}_j$.
\end{lemma}

\begin{proof}

First we assume that the circuit $V(\boldsymbol{\theta})$ consists of $p$ parameters, and could be written in the sequence:
\begin{equation*}
V(\boldsymbol{\theta}) = V_{p} (\theta_p) \cdot V_{p-1} (\theta_{p-1}) \cdots V_{1} (\theta_1),	
\end{equation*}
such that each block $V_j$ contains only one parameter.

The parameter shifting rule \cite{PhysRevA.98.032309, PhysRevLett.118.150503} provides a gradient calculation method for the single parameter case.  Consider the observable $O' = V_{j+1}^{\dag} \cdots V_p^\dag O V_p \cdots V_{j+1}$ and the input state  $\rho_{\text{in}}' = V_{j-1} \cdots V_1 \rho_{\text{in}} V_1^\dag \cdots V_{j-1}^{\dag}$. The gradient of $f_j (\theta_j)= \Tr[O' \cdot U(\theta_j) \rho_{\text{in}}' U(\theta_j)^{\dag}]$ could be calculated as
\begin{equation*}
\frac{d f_j}{d \theta_j} = f_j(\theta_j+\frac{\pi}{4}) - f_j(\theta_j - \frac{\pi}{4}).	
\end{equation*}
Thus, by inserting the form of $O'$ and $\rho_{\text{in}}'$, we could obtain
\begin{align*}
&\frac{\partial f}{\partial \theta_j} = \frac{d f_j}{d \theta_j} = f_j(\theta_j+\frac{\pi}{4}) - f_j(\theta_j - \frac{\pi}{4}) \\
&= \Tr [O \cdot V(\boldsymbol{\theta}_+) \rho_{\text{in}} V(\boldsymbol{\theta}_+)^{\dag}] - \Tr [O \cdot V(\boldsymbol{\theta}_-) \rho_{\text{in}} V(\boldsymbol{\theta}_-)^{\dag}].
\end{align*}
\end{proof}

\begin{lemma}\label{tqnn_lemma_wawbwcwd}
Let $\theta$ be a variable with uniform distribution in $[0,2\pi]$. Let $G$ be {an arbitrary hermitian unitary}, and define $W=e^{-i\theta G}$.
Let $O$ be an arbitrary quantum observable. 
Then
\begin{align}
&\mathop{\E}_{\theta} {\rm Tr}[O W \rho_1 W^\dag] {\rm Tr} [O W \rho_2 W^\dag] = {\rm Tr} [O_1 \rho_1] {\rm Tr} [O_1 \rho_2] \nonumber \\
&+ \frac{1}{2} {\rm Tr} [O_2 \rho_1] {\rm Tr} [O_2 \rho_2] + \frac{1}{2} {\rm Tr} [iO_2 G \rho_1] {\rm Tr} [iO_2 G \rho_2], \label{tqnn_lemma_wawbwcwd_eq}
\end{align}
where $\rho_1,\rho_2$ are two quantum states, and $O_{1,2}= \frac{1}{2} O \pm \frac{1}{2} GOG$, respectively. 
\end{lemma}

\begin{proof}

By replacing the term 
\begin{equation*}
W=e^{-i \theta G}= I \cos \theta -i G \sin \theta,
\end{equation*}
we have
\begin{align}
&\ \quad {\rm Tr}[O W \rho_1 W^\dag] \nonumber \\
&= {\rm Tr} [O (I \cos \theta -i G \sin \theta) \rho_1 (I \cos \theta +i G \sin \theta) ] \nonumber \\
&= \Tr [O_1 \rho_1] + \Tr [O_2 \rho_1] \cos 2\theta - \sin 2\theta \Tr [iO_2 G \rho_1], \label{tqnn_lemma_wawbwcwd_1_3}
\end{align}
where hermitians $O_1$ and $O_2$ are commuting and anti-commuting parts of $O$ with respect to $G$, respectively, i.e.
\begin{align*}
O_1 + O_2 &= O,\\
O_1 G - G O_1 &= \frac{OG+GO}{2} - \frac{GO+OG}{2} = 0,\\
O_2 G + G O_2 &= \frac{OG-GO}{2} + \frac{GO-OG}{2} = 0. 	
\end{align*}
We remark that $i O_2 G$ could serve as a hermitian observable. Similar formulation also holds for $\rho_2$, i.e.
\begin{align}
&\ \quad {\rm Tr}[O W \rho_2 W^\dag] \nonumber \\
&= \Tr [O_1 \rho_2] + \Tr [O_2 \rho_2] \cos 2\theta - \sin 2\theta \Tr [iO_2 G \rho_2].\label{tqnn_lemma_wawbwcwd_2_2}
\end{align}
Combining Eqs.~(\ref{tqnn_lemma_wawbwcwd_1_3}) and (\ref{tqnn_lemma_wawbwcwd_2_2}), we have
\begin{align}
&\ \quad \mathop{\E}_{\theta} {\rm Tr}[O W \rho_1 W^\dag] {\rm Tr} [O W \rho_2 W^\dag] \nonumber \\
&= \mathop{\E}_{\theta} \left( \Tr [O_1 \rho_1] + \Tr [O_2 \rho_1] \cos 2\theta - \sin 2\theta \Tr [iO_2 G \rho_1] \right) \nonumber \\
&\ \quad \cdot \left( \Tr [O_1 \rho_2] + \Tr [O_2 \rho_2] \cos 2\theta - \sin 2\theta \Tr [iO_2 G \rho_2] \right) \nonumber \\
&= \Tr [O_1 \rho_1] \Tr [O_1 \rho_2] + \frac{1}{2} \Tr [O_2 \rho_1] \Tr [O_2 \rho_2] \nonumber \\
&\qquad + \frac{1}{2} \Tr [iO_2 G \rho_1] \Tr [iO_2 G \rho_2], \label{tqnn_lemma_wawbwcwd_3_5}
\end{align}
where Eq.~(\ref{tqnn_lemma_wawbwcwd_3_5}) is obtained by calculating expectation terms.

\end{proof}

\begin{lemma}\label{tqnn_lemma_w_plus_minus}
Let $\theta$ be a variable with uniform distribution in $[0,2\pi]$. Let $G$ be {an arbitrary hermitian unitary}, and denote $W_{\pm}=e^{-i(\theta \pm \frac{\pi}{4}) G}$.  Let $O$ be an arbitrary hermitian observable and . 
Then
\begin{multline}
\E_{\theta} \left( {\rm Tr} [O W_+ \rho W_{+}^{\dag}] - {\rm Tr} [O W_- \rho W_-^\dag ] \right)^2 \\
=2 {\rm Tr}[O_2 \rho]^2 + 2 {\rm Tr}[iO_2 G \rho]^2, 
\end{multline}
where $\rho$ is a quantum state and $O_{2}= \frac{1}{2} O - \frac{1}{2} GOG $. 
\end{lemma}

\begin{proof}

First we denote $\theta_\pm = \theta \pm \frac{\pi}{4}$ for convenience.
By replacing terms
\begin{equation*}
W_{\pm}=e^{-i\theta_\pm G} = I \cos \theta_\pm  - i G \sin \theta_\pm,
\end{equation*}
we obtain
\begin{align*}
{}&{} {\rm Tr} [O W_+ \rho W_{+}^{\dag}] - {\rm Tr} [O W_- \rho W_-^\dag ] \\
={}&{} \Tr [O (I\cos \theta_+ -iG \sin \theta_+) \rho (I\cos \theta_+ +iG \sin \theta_+)] \\
-{}&{} \Tr [O (I\cos \theta_- -iG \sin \theta_-) \rho (I\cos \theta_- +iG \sin \theta_-)] \\
={}&{} \Tr[O_1 \rho] + \Tr[O_2 \rho] \cos 2\theta_+ - \sin 2\theta_+ \Tr[iO_2 G \rho] \\
-{}&{} \Tr[O_1 \rho] - \Tr[O_2 \rho] \cos 2\theta_- + \sin 2\theta_- \Tr[iO_2 G \rho] \\
={}&{} -2 \sin 2\theta \Tr[O_2 \rho] -2\cos 2\theta \Tr[iO_2 G \rho],
\end{align*}
where hermitians $O_{1,2}= \frac{1}{2} O \pm \frac{1}{2} GOG$ are commuting and anti-commuting parts of $O$ with respect to $G$, respectively, i.e.
\begin{align*}
O_1 + O_2 &= O,\\
O_1 G - G O_1 &= \frac{OG+GO}{2} - \frac{GO+OG}{2} = 0,\\
O_2 G + G O_2 &= \frac{OG-GO}{2} + \frac{GO-OG}{2} = 0. 	
\end{align*}
We remark that $i O_2 G$ could be served as a hermitian observable.
Thus, we have
\begin{align}
{}&{}\E_{\theta} \left( {\rm Tr} [O W_+ \rho W_{+}^{\dag}] - {\rm Tr} [O W_- \rho W_-^\dag ] \right)^2 \nonumber \\
={}&{} 2 \Tr[O_2 \rho]^2 + 2 \Tr[iO_2 G \rho]^2, \label{tqnn_lemma_w_plus_minus_2_2}
\end{align}
where Eq.~(\ref{tqnn_lemma_w_plus_minus_2_2}) is obtained by calculating expectation terms.

\end{proof}

\section{Proof of Theorem~\ref{tqnn_appendix_lemma_cost}}

\label{tqnn_app_proof_cost}

\begin{proof}
We first introduce several notations for convenience. 
\medskip
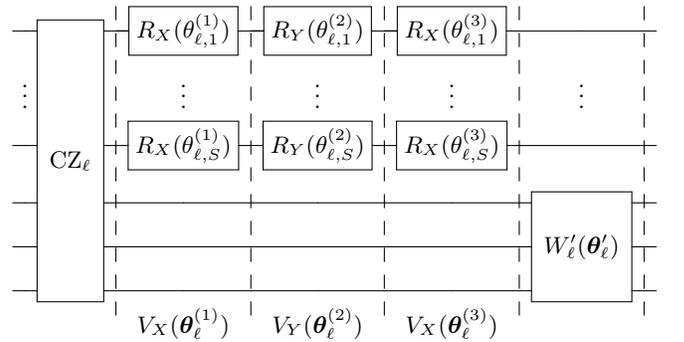
\begin{figure}[H]
\centerline{
\Qcircuit @C=0.5em @R=0.9em {
\lstick{} & \qw & \multigate{5}{{\rm CZ}_\ell} \barrier[-0em]{5} & \qw &  \gate{R_X(\theta_{\ell,1}^{(1)})} \qw \barrier[-0em]{5} & \qw & \gate{R_Y(\theta_{\ell,1}^{(2)})} \qw \barrier[-0em]{5} & \qw & \gate{R_X(\theta_{\ell,1}^{(3)})} \qw \barrier[-0em]{5} & \qw & \qw \barrier[-0em]{5} & \qw & \qw \\
\lstick{} & {\vdots} & \nghost{{\rm CZ}_\ell} &  & {\vdots} &  & {\vdots} &  & {\vdots} &  & {\vdots} &  &  \\
\lstick{} & \qw & \ghost{{\rm CZ}_\ell} & \qw & \gate{R_X(\theta_{\ell,S}^{(1)})} & \qw & \gate{R_Y(\theta_{\ell,S}^{(2)})} & \qw & \gate{R_X(\theta_{\ell,S}^{(3)})} & \qw & \qw & \qw & \qw \\
\lstick{} & \qw & \ghost{{\rm CZ}_\ell} & \qw & \qw & \qw & \qw & \qw & \qw & \qw & \multigate{2}{W'_\ell (\boldsymbol{\theta}_\ell')} & \qw & \qw \\
\lstick{} & \qw & \ghost{{\rm CZ}_\ell} & \qw & \qw & \qw & \qw & \qw & \qw & \qw & \ghost{W'_\ell (\boldsymbol{\theta}_\ell')} & \qw & \qw \\
\lstick{} & \qw & \ghost{{\rm CZ}_\ell} & \qw & \qw & \qw & \qw & \qw & \qw & \qw & \ghost{W'_\ell (\boldsymbol{\theta}_\ell')} & \qw & \qw \\
& & & & V_X (\boldsymbol{\theta}_\ell^{(1)}) & & V_Y (\boldsymbol{\theta}_\ell^{(2)}) & & V_X (\boldsymbol{\theta}_\ell^{(3)}) & & & & 
}
}
\caption{Decomposition of a controlled layer block into five layers.}
\label{tqnn_appendix_circuit}
\end{figure}

Denote 
\begin{align}
V_{X}(\boldsymbol{\theta}_{\ell}^{(1)}) &\equiv R_X(\theta_{\ell,1}^{(1)}) \otimes \cdots \otimes R_X(\theta_{\ell,S}^{(1)})\otimes \sigma_0 \otimes \cdots \otimes \sigma_0, \label{tqnn_appendix_lemma_cost_1_1}\\
V_{Y}(\boldsymbol{\theta}_{\ell}^{(2)}) &\equiv R_Y(\theta_{\ell,1}^{(2)}) \otimes \cdots \otimes R_Y(\theta_{\ell,S}^{(2)}) \otimes \sigma_0 \otimes \cdots \otimes \sigma_0, \label{tqnn_appendix_lemma_cost_1_2}\\
V_{X}(\boldsymbol{\theta}_{\ell}^{(3)}) &\equiv R_X(\theta_{\ell,1}^{(3)}) \otimes \cdots \otimes R_X(\theta_{\ell,S}^{(3)}) \otimes \sigma_0 \otimes \cdots \otimes \sigma_0,  \label{tqnn_appendix_lemma_cost_1_3}
\end{align}
where $\boldsymbol{\theta}_\ell^{(i)}=(\theta_{\ell,1}^{(i)},\cdots,\theta_{\ell,S}^{(i)}),\ \forall i \in \{1,2,3\}$.
Then a controlled-layer block in Figure~\ref{tqnn_appendix_circuit} could be divided into five parts as follows:
\begin{equation}\label{tqnn_appendix_lemma_cost_gate_v}
V_\ell(\boldsymbol{\theta}_\ell) = (I^{\otimes S} \otimes W_\ell ' (\boldsymbol{\theta}_\ell ')) V_X (\boldsymbol{\theta}_\ell^{(3)}) V_Y (\boldsymbol{\theta}_\ell^{(2)}) V_X (\boldsymbol{\theta}_\ell^{(1)}) {\rm CZ}_\ell,	
\end{equation}
where the parameter $\boldsymbol{\theta}_\ell = (\boldsymbol{\theta}_\ell^{(1)}, \boldsymbol{\theta}_\ell^{(2)}, \boldsymbol{\theta}_\ell^{(3)}, \boldsymbol{\theta}_\ell ')$.
We abbreviate  the term $I^{\otimes S} \otimes W_\ell ' (\boldsymbol{\theta}_\ell ')$ as 
 $W_\ell ' (\boldsymbol{\theta}_\ell ')$ in the following for convenience.

In the following proof, we consider $V_{X}(\boldsymbol{\theta}_{\ell}^{(1)})$, $V_{Y}(\boldsymbol{\theta}_{\ell}^{(2)})$, $V_{X}(\boldsymbol{\theta}_{\ell}^{(3)})$, $W_\ell ' (\boldsymbol{\theta}_{\ell} ')$ and ${\rm CZ}_\ell$ as separate layers.
We denote by $\rho_{k}$ the state after the $k$th layer, $\forall k \in \{0,1,\cdots,5L\}$,
\begin{equation}\label{tqnn_appendix_lemma_cost_rho_k}
\rho_k := \left\{
\begin{aligned}
& \left( \prod_{i=\ell}^{1} V_i (\boldsymbol{\theta}_i) \right) \rho_{\text{in}} \left( \prod_{i=1}^{\ell} V_i (\boldsymbol{\theta}_i)^\dag \right) &\ (k=5\ell), \\
& {\rm CZ}_{\ell+1} \rho_{5\ell} {\rm CZ}_{\ell+1}^{\dag} &\ (k=5\ell+1), \\
& V_{X}(\boldsymbol{\theta}_{\ell+1}^{(1)}) \rho_{5\ell+1} V_{X}(\boldsymbol{\theta}_{\ell+1}^{(1)})^{\dag} &\ (k=5\ell+2), \\ 
& V_{Y}(\boldsymbol{\theta}_{\ell+1}^{(2)}) \rho_{5\ell+2} V_{Y}(\boldsymbol{\theta}_{\ell+1}^{(2)})^{\dag} &\ (k=5\ell+3), \\ 
& V_{X}(\boldsymbol{\theta}_{\ell+1}^{(3)}) \rho_{5\ell+3} V_{X}(\boldsymbol{\theta}_{\ell+1}^{(3)})^{\dag} &\ (k=5\ell+4).
\end{aligned}
\right.
\end{equation}

Now we begin to prove Eq.~(\ref{tqnn_appendix_lemma_cost_equation}). We denote all parameters in the circuit by $\boldsymbol{\theta}=(\boldsymbol{\theta}_1, \boldsymbol{\theta}_2, \cdots, \boldsymbol{\theta}_L)$. Rewrite the formulation of $f(\boldsymbol{\theta})$ in detail:
\begin{align}
&\ \quad \mathop{\E}\limits_{\boldsymbol{\theta}} \left( \Tr \left[ \sigma_{\boldsymbol{i}} V(\boldsymbol{\theta}) \rho_{\text{in}} V(\boldsymbol{\theta})^{\dag}\right] \right)^2 \notag\\
&= \mathop{\E}\limits_{\boldsymbol{\theta}} \left( \Tr \left[ \sigma_{\boldsymbol{i}} \rho_{5L}  \right] \right)^2 \label{tqnn_appendix_lemma_cost_2_1} \\
&= \mathop{\E}\limits_{\boldsymbol{\theta}_1, \cdots, \boldsymbol{\theta}_{L-1}, \boldsymbol{\theta}_{L}^{(1)}, \boldsymbol{\theta}_{L}^{(2)}, \boldsymbol{\theta}_{L}^{(3)}, \boldsymbol{\theta}_{L}'} \Tr \left[ \sigma_{\boldsymbol{i}} W_{L}'(\boldsymbol{\theta}_L ')  \rho_{5L-1} W_{L}'(\boldsymbol{\theta}_L ')^\dag \right]^2 \label{tqnn_appendix_lemma_cost_2_2} \\
&= \mathop{\E}\limits_{\boldsymbol{\theta}_1, \cdots, \boldsymbol{\theta}_{L-1}, \boldsymbol{\theta}_{L}^{(1)}, \boldsymbol{\theta}_{L}^{(2)}, \boldsymbol{\theta}_{L}^{(3)} } \Tr \left[ \sigma_{\boldsymbol{i}} \rho_{5L-1} \right]^2 \label{tqnn_appendix_lemma_cost_2_3} \\
&= \mathop{\E}\limits_{\boldsymbol{\theta}_1, \cdots, \boldsymbol{\theta}_{L-1}, \boldsymbol{\theta}_{L}^{(1)}, \boldsymbol{\theta}_{L}^{(2)}, \boldsymbol{\theta}_{L}^{(3)}} \Tr \left[ \sigma_{\boldsymbol{i}} V_{X}(\boldsymbol{\theta}_{L}^{(3)}) \rho_{5L-2} V_{X}(\boldsymbol{\theta}_{L}^{(3)})^\dag \right]^2.
\label{tqnn_appendix_lemma_cost_2_4} 
\end{align}
Eqs.~(\ref{tqnn_appendix_lemma_cost_2_1}) and (\ref{tqnn_appendix_lemma_cost_2_2}) follow from the definition of $\rho_{k}$ in Eq.~(\ref{tqnn_appendix_lemma_cost_rho_k}).
Eq.~(\ref{tqnn_appendix_lemma_cost_2_3}) holds since
\begin{equation}\label{tqnn_appendix_lemma_cost_w_trival}
W_{L}'(\boldsymbol{\theta}_L ')^\dag \sigma_{\boldsymbol{i}} W_{L}'(\boldsymbol{\theta}_L ') =  \sigma_{\boldsymbol{i}}.
\end{equation}
Eq.~(\ref{tqnn_appendix_lemma_cost_2_4}) is obtained by using the definition of $\rho_{k}$ in Eq.~(\ref{tqnn_appendix_lemma_cost_rho_k}) again. 

Continuing from Eq.~(\ref{tqnn_appendix_lemma_cost_2_4}), we have 
\begin{align}
&\ \quad \mathop{\E}\limits_{\boldsymbol{\theta}} \left( \Tr \left[ \sigma_{\boldsymbol{i}} V(\boldsymbol{\theta}) \rho_{\text{in}} V(\boldsymbol{\theta})^{\dag}\right] \right)^2\notag \\
&\geq \mathop{\E}\limits_{\boldsymbol{\theta}_1, \cdots, \boldsymbol{\theta}_{L-1}, \boldsymbol{\theta}_{L}^{(1)}, \boldsymbol{\theta}_{L}^{(2)}} \frac{\Tr \left[ \sigma_{\boldsymbol{3|i;2}} V_{Y}(\boldsymbol{\theta}_{L}^{(2)}) \rho_{5L-3} V_{Y}(\boldsymbol{\theta}_{L}^{(2)})^\dag \right]^2}{2^S} \label{tqnn_appendix_lemma_cost_2_5} \\
&\geq \mathop{\E}\limits_{\boldsymbol{\theta}_1, \cdots, \boldsymbol{\theta}_{L-1}, \boldsymbol{\theta}_{L}^{(1)} } \frac{ \Tr \left[ \sigma_{\boldsymbol{3|i}} V_{X}(\boldsymbol{\theta}_{L}^{(1)}) \rho_{5L-4} V_{X}(\boldsymbol{\theta}_{L}^{(1)})^\dag \right]^2 }{2^{2S}} \label{tqnn_appendix_lemma_cost_2_6} \\
&\geq \left(\frac{1}{2} \right)^{3S} \mathop{\E}\limits_{\boldsymbol{\theta}_1, \cdots, \boldsymbol{\theta}_{L-1}} \Tr \left[ \sigma_{\boldsymbol{3|i}} \rho_{5L-4} \right]^2 \label{tqnn_appendix_lemma_cost_2_7} \\
&= \left(\frac{1}{2} \right)^{3S} \mathop{\E}\limits_{\boldsymbol{\theta}_1, \cdots, \boldsymbol{\theta}_{L-1}} \Tr \left[ \sigma_{\boldsymbol{3|i}} \rho_{5L-5} \right]^2, \label{tqnn_appendix_lemma_cost_2_8}
\end{align}
where $\boldsymbol{3|i;2}$ denotes the index by replacing non-zero elements of $\boldsymbol{i}=(i_1,\cdots,i_N)$ with $3$ if the original value is $2$. Eq.~(\ref{tqnn_appendix_lemma_cost_2_5}) is derived by using Lemma~\ref{tqnn_lemma_wawbwcwd} for the $R_X$ rotation case for $S$ times. 
Eqs.~(\ref{tqnn_appendix_lemma_cost_2_6}) and (\ref{tqnn_appendix_lemma_cost_2_7}) are derived similarly by using Lemma~\ref{tqnn_lemma_wawbwcwd}.
Eq.~(\ref{tqnn_appendix_lemma_cost_2_8}) is derived by noticing that for $j,k \in \{0,3\}$, the term $\sigma_j \otimes \sigma_k$ remains the same after the CZ operation.

Now we move forward from Eq.~(\ref{tqnn_appendix_lemma_cost_2_8}). 
By applying Eqs.~(\ref{tqnn_appendix_lemma_cost_2_2})--(\ref{tqnn_appendix_lemma_cost_2_8}) inductively for $L-1$ times, we obtain 
\begin{align}
\text{Eq.~(\ref{tqnn_appendix_lemma_cost_2_8})} \geq &{} \left(\frac{1}{2} \right)^{3LS} {\Tr} \left[ \sigma_{\boldsymbol{3|i}} \rho_{\text{in}} \right]^2 \label{tqnn_appendix_lemma_cost_3_5}.
\end{align}

Thus, we prove Equation~(\ref{tqnn_appendix_lemma_cost_equation}).

\end{proof}

\section{Proof of Theorem~\ref{tqnn_appendix_theorem}}

\label{tqnn_app_proof_gradient}

\begin{proof}

First, we notice that the norm of the whole gradient is lower bounded by that of particle derivatives summed over a part of parameters, i.e. 
\begin{align}
\mathop{\E}\limits_{\boldsymbol{\theta}} \|\nabla_{\boldsymbol{\theta}} f \|^2 {} & \geq 
\sum_{\ell=1}^{L-1} \sum_{n=1}^{S} \sum_{j=1}^{3}  \mathop{\E}\limits_{\boldsymbol{\theta}} \left(\frac{\partial f }{\partial \theta_{\ell,n}^{(j)}}\right)^2 , \label{tqnn_ttn_proof_2}
\end{align}
where $\theta_{\ell,n}^{(j)}$ denotes the parameter of the $j$th single-qubit gate on the $n$th qubit in the $\ell$th CL block.
Thus, we could obtain Eq.~(\ref{tqnn_appendix_theorem_eq}) if 
\begin{equation}\label{tqnn_appendix_lemma_partial_derivative_equation}
\mathop{\E}\limits_{\boldsymbol{\theta}} \left( \frac{\partial f}{\partial \theta_{\ell,n}^{(j)} } \right)^2 \geq \frac{4 \left( {\rm Tr}[\sigma_{\boldsymbol{3|i}} \rho_{\text{in}}] \right)^2 }{8^{LS}}
\end{equation}
holds for all $\ell \in \{1,\cdots,L-1\}$, $n \in \{1,\cdots,S\}$, and $j \in \{1,2,3\}$. 

Now we begin to prove Eq.~(\ref{tqnn_appendix_lemma_partial_derivative_equation}).
Our main idea is to integrate the square of the partial derivative of $f$ with respect to $\boldsymbol{\theta}=(\boldsymbol{\theta}_1,\cdots,\boldsymbol{\theta}_{L})$ by using Lemma~\ref{tqnn_lemma_wawbwcwd} and Lemma~\ref{tqnn_lemma_w_plus_minus}.
For convenience, we follow the definition of $V_{X}(\cdot)$ in Eqs.~(\ref{tqnn_appendix_lemma_cost_1_1}), (\ref{tqnn_appendix_lemma_cost_1_3}), $V_{Y}(\cdot)$ in Eq.~(\ref{tqnn_appendix_lemma_cost_1_2}), ${\rm CZ}_\ell$ and $\rho_{k}$ in Eq.~(\ref{tqnn_appendix_lemma_cost_rho_k}).
In addition, we denote
\begin{align}
&{}\ \quad \rho_{k,\ell,n,j,\pm} := \nonumber \\
&{} \left\{
\begin{aligned}
& \rho_k (\boldsymbol{\theta}_1,  \cdots)  &\ {(k \leq 5\ell-4)}, \\
& \rho_k (\boldsymbol{\theta}_1, \cdots, \boldsymbol{\theta}_{\ell-1}, \boldsymbol{\theta}_{\ell,n,j,\pm}, \boldsymbol{\theta}_{\ell+1}, \cdots )  &\ {(k > 5\ell-4)}, 
\end{aligned}
\right.\label{tqnn_appendix_lemma_partial_derivative_upper_bound_rho_k_l_n}
\end{align}
where $\rho_k$ follows from the definition in Eq.~(\ref{tqnn_appendix_lemma_cost_rho_k}), and $\boldsymbol{\theta}_{\ell,n,j,\pm}$ differs from $\boldsymbol{\theta}_\ell$ by plus or minus $\frac{\pi}{4}$ on the  component $\theta_{\ell,n}^{(j)}$.

Next, we rewrite the formulation in Eq.~(\ref{tqnn_appendix_lemma_partial_derivative_equation}),
\begin{align}
&{}\ \quad \mathop{\E}\limits_{\boldsymbol{\theta}} \left( \frac{\partial }{\partial \theta_{\ell,n}^{(j)}} {\rm Tr} \left[ \sigma_{\boldsymbol{i}} V(\boldsymbol{\theta}) \rho_{\text{in}} V(\boldsymbol{\theta})^{\dag}\right] \right)^2 \label{tqnn_appendix_lemma_partial_derivative_bound_1_1} \\
&= \mathop{\E}\limits_{\boldsymbol{\theta}_1} \cdots \mathop{\E}\limits_{\boldsymbol{\theta}_L} \left( \frac{\partial }{\partial \theta_{\ell,n}^{(j)}} {\rm Tr} \left[ \sigma_{\boldsymbol{i}} \rho_{5L} \right] \right)^2 \label{tqnn_appendix_lemma_partial_derivative_bound_1_2} \\
&= \mathop{\E}\limits_{\boldsymbol{\theta}_1} \cdots \mathop{\E}\limits_{\boldsymbol{\theta}_L} \left( {\rm Tr} \left[ \sigma_{\boldsymbol{i}} \rho_{5L,\ell,n,j,+} \right] -  {\rm Tr} \left[ \sigma_{\boldsymbol{i}} \rho_{5L,\ell,n,j,-} \right] \right)^2. \label{tqnn_appendix_lemma_partial_derivative_bound_1_3}
\end{align}
Eq.~(\ref{tqnn_appendix_lemma_partial_derivative_bound_1_1}) follows from the formulation of the cost function $f$. 
Eq.~(\ref{tqnn_appendix_lemma_partial_derivative_bound_1_2}) follows from Eq.~(\ref{tqnn_appendix_lemma_cost_rho_k}). 
Eq.~(\ref{tqnn_appendix_lemma_partial_derivative_bound_1_3}) follows from the parameter-shift rule in Lemma~\ref{tqnn_params_shifting}. 

We proceed from Eq.~(\ref{tqnn_appendix_lemma_partial_derivative_bound_1_3}).
Firstly, by taking the expectation with the uniform distribution of parameters in $\boldsymbol{\theta}_{L}$, we obtain
\begin{align}
&\mathop{\E}\limits_{\boldsymbol{\theta}_L} \left( {\rm Tr} \left[ \sigma_{\boldsymbol{i}} \rho_{5L,\ell,n,j,+} \right] -  {\rm Tr} \left[ \sigma_{\boldsymbol{i}} \rho_{5L,\ell,n,j,-} \right] \right)^2  \notag\\
={}& \mathop{\E}\limits_{\boldsymbol{\theta}_{L}^{(1)}} \mathop{\E}\limits_{\boldsymbol{\theta}_{L}^{(2)}} \mathop{\E}\limits_{\boldsymbol{\theta}_{L}^{(3)}}  \mathop{\E}\limits_{\boldsymbol{\theta}_{L}'} \Big( {\rm Tr} \big[ \sigma_{\boldsymbol{i}} W_L'(\boldsymbol{\theta}_L')  \rho_{5L-1,\ell,n,j,+} W_L'(\boldsymbol{\theta}_L')^{\dag} \big] \nonumber \\
&\quad -  {\rm Tr} \big[ \sigma_{\boldsymbol{i}} W_L'(\boldsymbol{\theta}_L')  \rho_{5L-1,\ell,n,j,-} W_L'(\boldsymbol{\theta}_L')^{\dag}  \big] \Big)^2  \label{tqnn_appendix_lemma_partial_derivative_2_1} \\
={}& \mathop{\E}\limits_{\boldsymbol{\theta}_{L}^{(1)}} \mathop{\E}\limits_{\boldsymbol{\theta}_{L}^{(2)}} \mathop{\E}\limits_{\boldsymbol{\theta}_{L}^{(3)}} \Big( {\rm Tr} \big[ \sigma_{\boldsymbol{i}} \rho_{5L-1,\ell,n,j,+} \big] \nonumber \\
&\quad -  {\rm Tr} \big[ \sigma_{\boldsymbol{i}}  \rho_{5L-1,\ell,n,j,-} \big] \Big)^2  \label{tqnn_appendix_lemma_partial_derivative_2_2} \\
={}& \mathop{\E}\limits_{\boldsymbol{\theta}_{L}^{(1)}} \mathop{\E}\limits_{\boldsymbol{\theta}_{L}^{(2)}} \mathop{\E}\limits_{\boldsymbol{\theta}_{L}^{(3)}}  \Big( {\rm Tr} \big[ \sigma_{\boldsymbol{i}} V_{X}(\boldsymbol{\theta}_{L}^{(3)}) \rho_{5L-2,\ell,n,j,+} V_{X}(\boldsymbol{\theta}_{L}^{(3)})^\dag \big] \nonumber \\
&\quad -  {\rm Tr} \big[ \sigma_{\boldsymbol{i}} V_{X}(\boldsymbol{\theta}_{L}^{(3)})  \rho_{5L-2,\ell,n,j,-} V_{X}(\boldsymbol{\theta}_{L}^{(3)})^\dag \big] \Big)^2,  \label{tqnn_appendix_lemma_partial_derivative_2_3} 
\end{align}
where Eq.~(\ref{tqnn_appendix_lemma_partial_derivative_2_1}) follows from the definition of $\rho_{k,\ell,n,j,\pm}$ in Eq.~(\ref{tqnn_appendix_lemma_partial_derivative_upper_bound_rho_k_l_n}). 
Eq.~(\ref{tqnn_appendix_lemma_partial_derivative_2_2}) is obtained by using Eq.~(\ref{tqnn_appendix_lemma_cost_w_trival}).
Eq.~(\ref{tqnn_appendix_lemma_partial_derivative_2_3}) follows from the definition of $\rho_{k,\ell,n,j,\pm}$ in Eq.~(\ref{tqnn_appendix_lemma_partial_derivative_upper_bound_rho_k_l_n}) again.

\medskip
Then, we have
\begin{align}
\geq{}& \frac{1}{2^S} \mathop{\E}\limits_{\boldsymbol{\theta}_{L}^{(1)}} \mathop{\E}\limits_{\boldsymbol{\theta}_{L}^{(2)}}  \Big( {\rm Tr} \big[ \sigma_{\boldsymbol{3|i;2}} \rho_{5L-2,\ell,n,j,+}  \big] \nonumber \\
&\quad -  {\rm Tr} \big[ \sigma_{\boldsymbol{3|i;2}} \rho_{5L-2,\ell,n,j,-} \big] \Big)^2  \label{tqnn_appendix_lemma_partial_derivative_2_4} \\
={}& \frac{1}{2^S} \mathop{\E}\limits_{\boldsymbol{\theta}_{L}^{(1)}} \mathop{\E}\limits_{\boldsymbol{\theta}_{L}^{(2)}}  \Big( {\rm Tr} \big[ \sigma_{\boldsymbol{3|i;2}} V_Y(\boldsymbol{\theta}_{L}^{(2)}) \rho_{5L-3,\ell,n,j,+} V_Y(\boldsymbol{\theta}_{L}^{(2)})^\dag \big] \nonumber \\
&\quad -  {\rm Tr} \big[ \sigma_{\boldsymbol{3|i;2}} V_Y(\boldsymbol{\theta}_{L}^{(2)}) \rho_{5L-3,\ell,n,j,-} V_Y(\boldsymbol{\theta}_{L}^{(2)})^\dag \big] \Big)^2  \label{tqnn_appendix_lemma_partial_derivative_2_5} \\
\geq{}& \frac{\mathop{\E}\limits_{\boldsymbol{\theta}_{L}^{(1)}} \Big( {\rm Tr} \big[ \sigma_{\boldsymbol{3|i}} \rho_{5L-3,\ell,n,j,+} \big] -  {\rm Tr} \big[ \sigma_{\boldsymbol{3|i}} \rho_{5L-3,\ell,n,j,-} \big] \Big)^2}{2^{2S}} \label{tqnn_appendix_lemma_partial_derivative_2_6} \\
={}& \frac{1}{2^{2S}} \mathop{\E}\limits_{\boldsymbol{\theta}_{L}^{(1)}} \Big( {\rm Tr} \big[ \sigma_{\boldsymbol{3|i}} V_X(\boldsymbol{\theta}_{L}^{(1)}) \rho_{5L-4,\ell,n,j,+} V_X(\boldsymbol{\theta}_{L}^{(1)})^\dag \big] \nonumber \\
&\quad -  {\rm Tr} \big[ \sigma_{\boldsymbol{3|i}} V_X(\boldsymbol{\theta}_{L}^{(1)}) \rho_{5L-4,\ell,n,j,-} V_X(\boldsymbol{\theta}_{L}^{(1)})^\dag \big] \Big)^2  \label{tqnn_appendix_lemma_partial_derivative_2_7} \\
\geq{}& \frac{\Big( {\rm Tr} \big[ \sigma_{\boldsymbol{3|i}} \rho_{5L-4,\ell,n,j,+} \big] -  {\rm Tr} \big[ \sigma_{\boldsymbol{3|i}} \rho_{5L-4,\ell,n,j,-} \big] \Big)^2 }{2^{3S}}  \label{tqnn_appendix_lemma_partial_derivative_2_8} \\
={}& \frac{\Big( {\rm Tr} \big[ \sigma_{\boldsymbol{3|i}} \rho_{5L-5,\ell,n,j,+} \big] -  {\rm Tr} \big[ \sigma_{\boldsymbol{3|i}} \rho_{5L-5,\ell,n,j,-} \big] \Big)^2 }{2^{3S}}  \label{tqnn_appendix_lemma_partial_derivative_2_9},
\end{align}
where $\boldsymbol{3|i;2}$ denotes the index by replacing non-zero elements of $\boldsymbol{i}=(i_1,\cdots,i_N)$ with $3$ if the original value is $2$. 
Eq.~(\ref{tqnn_appendix_lemma_partial_derivative_2_4}) is obtained by using Lemma~\ref{tqnn_lemma_wawbwcwd} for the $R_X$ case. 
Eq.~(\ref{tqnn_appendix_lemma_partial_derivative_2_5}) follows from the definition of $\rho_{k,\ell,n,j,\pm}$ in Eq.~(\ref{tqnn_appendix_lemma_partial_derivative_upper_bound_rho_k_l_n}).
Eq.~(\ref{tqnn_appendix_lemma_partial_derivative_2_6}) is obtained by using Lemma~\ref{tqnn_lemma_wawbwcwd} for the $R_Y$ case. 
Eq.~(\ref{tqnn_appendix_lemma_partial_derivative_2_7}) follows from the definition of $\rho_{k,\ell,n,j,\pm}$ in Eq.~(\ref{tqnn_appendix_lemma_partial_derivative_upper_bound_rho_k_l_n}).
Eq.~(\ref{tqnn_appendix_lemma_partial_derivative_2_8}) is obtained by using Lemma~\ref{tqnn_lemma_wawbwcwd} for the $R_X$ case. 
Eq.~(\ref{tqnn_appendix_lemma_partial_derivative_2_9}) is derived by noticing that for $j,k \in \{0,3\}$, the term $\sigma_j \otimes \sigma_k$ remains the same after the CZ operation.

Next, we repeat the derivation in Eqs.(\ref{tqnn_appendix_lemma_partial_derivative_2_1})--(\ref{tqnn_appendix_lemma_partial_derivative_2_9}) inductively for parameters $(\boldsymbol{\theta}_{L-2}, \cdots, \boldsymbol{\theta}_{\ell+1})$, which yields
\begin{align}
{}&{} \mathop{\E}\limits_{\boldsymbol{\theta}} \left( \frac{\partial f}{\partial \theta_{\ell,n}^{(j)} } \right)^2 \nonumber \\
\geq{}& \frac{\mathop{\E}\limits_{\boldsymbol{\theta}_1} \cdots \mathop{\E}\limits_{\boldsymbol{\theta}_{\ell}} \Big( {\rm Tr} \big[ \sigma_{\boldsymbol{3|i}} \rho_{5\ell,\ell,n,j,+} \big] -  {\rm Tr} \big[ \sigma_{\boldsymbol{3|i}} \rho_{5\ell,\ell,n,j,-} \big] \Big)^2}{2^{3(L-\ell)S}}  \label{tqnn_appendix_lemma_partial_derivative_3_1} \\
={}& \frac{1}{8^{(L-\ell)S}} \mathop{\E}\limits_{\boldsymbol{\theta}_1} \cdots \mathop{\E}\limits_{\boldsymbol{\theta}_{\ell}}  \Big( {\rm Tr} \big[ \sigma_{\boldsymbol{3|i}} W_\ell ' (\boldsymbol{\theta}_\ell ')\rho_{5\ell-1,\ell,n,j,+} W_\ell ' (\boldsymbol{\theta}_\ell ')^\dag \big] \nonumber \\
&\quad -  {\rm Tr} \big[ \sigma_{\boldsymbol{3|i}} W_\ell ' (\boldsymbol{\theta}_\ell ')\rho_{5\ell-1,\ell,n,j,-} W_\ell ' (\boldsymbol{\theta}_\ell ')^\dag \big] \Big)^2  \label{tqnn_appendix_lemma_partial_derivative_3_4} \\
={}& \frac{1}{8^{(L-\ell)S}} \mathop{\E}\limits_{\boldsymbol{\theta}_1} \cdots \mathop{\E}\limits_{\boldsymbol{\theta}_{\ell-1}} \mathop{\E}\limits_{\boldsymbol{\theta}_{\ell}^{(1)}} \mathop{\E}\limits_{\boldsymbol{\theta}_{\ell}^{(2)}} \mathop{\E}\limits_{\boldsymbol{\theta}_{\ell}^{(3)}} \nonumber \\
&\quad \Big( {\rm Tr} \big[ \sigma_{\boldsymbol{3|i}} \rho_{5\ell-1,\ell,n,j,+}  \big] -  {\rm Tr} \big[ \sigma_{\boldsymbol{3|i}}  \rho_{5\ell-1,\ell,n,j,-} \big] \Big)^2 \label{tqnn_appendix_lemma_partial_derivative_3_5}.
\end{align}
Eq.~(\ref{tqnn_appendix_lemma_partial_derivative_3_4}) is derived by using the definition of $\rho_{k,\ell,n,j,\pm}$ in Eq.~(\ref{tqnn_appendix_lemma_partial_derivative_upper_bound_rho_k_l_n}).
Eq.~(\ref{tqnn_appendix_lemma_partial_derivative_3_5}) is derived by using Eq.~(\ref{tqnn_appendix_lemma_cost_w_trival}). 

Now we proceed to integrate parameter layers $(\boldsymbol{\theta}_\ell^{(1)}, \boldsymbol{\theta}_\ell^{(2)}, \boldsymbol{\theta}_\ell^{(3)})$ and $(\boldsymbol{\theta}_1, \cdots, \boldsymbol{\theta}_{\ell-1})$. Remark that the parameter $\theta_{\ell,n}^{(j)}$ is applied on the $n$th qubit with the observable $\sigma_3$ in Eq.~(\ref{tqnn_appendix_lemma_partial_derivative_3_5}). Since $\sigma_3$ anticommutes with $\sigma_1$ and $\sigma_2$, Hamiltonians of parameterized single-qubit gates in the circuit, we could apply Lemma~\ref{tqnn_lemma_wawbwcwd} and Lemma~\ref{tqnn_lemma_w_plus_minus} to simplify the formulation. Thus, {Eq.~(\ref{tqnn_appendix_lemma_partial_derivative_3_5})} could be further bounded
\begin{align}
\geq {}& \frac{\mathop{\E}\limits_{\boldsymbol{\theta}_1} \cdots \mathop{\E}\limits_{\boldsymbol{\theta}_{\ell-1}} 4  {\rm Tr} \big[ \sigma_{\boldsymbol{3|i}} \rho_{5\ell-5}  \big]^2}{8^{(L-\ell+1)S}}  \label{tqnn_appendix_lemma_partial_derivative_4_2} \\
\geq {}& \frac{4 {\rm Tr} \big[ \sigma_{\boldsymbol{3|i}} \rho_{\text{in}}  \big]^2}{8^{LS}} . \label{tqnn_appendix_lemma_partial_derivative_4_3}
\end{align}
Eq.~(\ref{tqnn_appendix_lemma_partial_derivative_4_2}) is obtained by noticing the similar collapsed formulation of Lemmas~\ref{tqnn_lemma_wawbwcwd} and \ref{tqnn_lemma_w_plus_minus} when the gate Hamiltonian anticommutes with the observable. 
Eq.~(\ref{tqnn_appendix_lemma_partial_derivative_4_3}) follows similar to Eq.(\ref{tqnn_appendix_lemma_cost_3_5}).
Thus, we prove Eq.~(\ref{tqnn_appendix_lemma_partial_derivative_equation}).

\end{proof}

\end{document}